\documentclass[letterpaper,twocolumn,10pt]{article}
\usepackage{usenix-2020-09}

\usepackage{filecontents}
\usepackage{tikz}
\usetikzlibrary{graphs,graphs.standard,quotes}

\usepackage{cite}
\usepackage[ruled,vlined]{algorithm2e}
\usepackage{amsmath,amssymb,amsfonts}
\usepackage{algorithmicx}
\usepackage{algpseudocode}
\usepackage{algorithm2e}
\usepackage{graphicx}
\usepackage{textcomp}
\usepackage{xcolor}
\usepackage{tabu}
\usepackage{tabularx}
\usepackage{pdflscape}
\usepackage[font={small,it}]{caption}
\usepackage{wrapfig}
\usepackage{subcaption}
\usepackage{mathtools}
\usepackage{amsthm}
\newtheorem{theorem}{Theorem}
\usepackage[toc,page]{appendix}
\usepackage{makecell}
\usepackage{hyperref}
\usepackage[toc,page]{appendix}

\usepackage{filecontents}
\usepackage{tikz}
\usetikzlibrary{graphs,graphs.standard,quotes}
\usepackage{amsmath,amssymb,amsfonts}
\usepackage[ruled,vlined]{algorithm2e}
\usepackage{algorithmicx}
\usepackage{algpseudocode}
\usepackage{algorithm2e}
\usepackage{graphicx}
\usepackage{textcomp}
\usepackage{xcolor}
\usepackage{tabu}
\usepackage{tabularx}
\usepackage{pdflscape}
\usepackage[font={small,it}]{caption}
\usepackage{wrapfig}
\usepackage{subcaption}
\usepackage{mathtools}
\usepackage{amsthm}
\usepackage{hyperref}
\usepackage[toc,page]{appendix}
\usepackage{makecell}
\usepackage{xcolor,colortbl}

\newcommand{\rom}[1]{\uppercase\expandafter{\romannumeral #1\relax}}
\SetKwRepeat{Do}{do}{while}

\begin{document}
\date{}

\title{\Large \bf SAFELearning: Enable Backdoor Detectability In Federated Learning With Secure Aggregation}

\author{
{\rm Zhuosheng Zhang}\\
Stevens Institute of Technology\\
zzhang97@stevens.edu

\and
{\rm Jiarui Li}\\
Stevens Institute of Technology\\
jli148@stevens.edu

\and
 {\rm Shucheng Yu}\\
Stevens Institute of Technology\\
shucheng.yu@stevens.edu

\and
 {\rm Christian Makaya}\\
 wchrismak@gmail.com
} 

\maketitle

\begin{abstract}
For model privacy, local model parameters in federated learning shall be obfuscated before sent to the remote aggregator. This technique is referred to as \emph{secure aggregation}. However, secure aggregation makes model poisoning attacks such backdooring more convenient considering that existing anomaly detection methods mostly require access to plaintext local models. This paper proposes SAFELearning which supports backdoor detection for secure aggregation. We achieve this through two new primitives - \emph{oblivious random grouping (ORG)} and \emph{partial parameter disclosure (PPD)}. ORG partitions participants into one-time random subgroups with group configurations oblivious to participants; PPD allows secure partial disclosure of aggregated subgroup models for anomaly detection without leaking individual model privacy. SAFELearning can significantly reduce backdoor model accuracy without jeopardizing the main task accuracy under common backdoor strategies. Extensive experiments show SAFELearning is robust against malicious and faulty participants, whilst being more efficient than the state-of-art secure aggregation protocol in terms of both communication and computation costs.
\end{abstract}




\maketitle

\section{Introduction}
Federated learning \cite{konevcny2016federated} becomes increasingly attractive in emerging applications \cite{Google-Fed,Google-Gboard} with the proliferation of Internet of Things and edge computing. As compared to centralized learning (i.e., training models at the central server), federated learning allows participants (i.e., users) to locally train models with their private data sets and only transmit the trained model parameters (or gradients) to the remote server. The latter aggregates local parameters to obtain a global model and returns it to users for next iteration. However, recent research has discovered that disclosing local models also poses threats to data privacy, either directly or under subtle attacks such as reconstruction attacks and model inversion attacks \cite{RC-2019}. To protect local models against disclosure, nobody except for the participant shall know her own local model while the global model will be revealed to all participants. This problem is known as \emph{Secure Aggregation} which can be generally realized using cryptographic primitives such as secure multiparty computation (SMC) and homomorphic encryption (HE). Despite of recent progresses, existing SMC or HE constructions are still not fully ready for complex networks (e.g., in deep learning) because of the high communication and/or computing overhead. Differential privacy \cite{chan2012privacy,shi2011privacy} provides efficient alternative solutions to model privacy protection. However, it remains a challenge to maintain an appropriate trade-off between privacy and model quality in deep learning tasks.   

Pairwise masking \cite{AC11, GX15} has recently caused attention for its efficiency in secure aggregation. Specifically, assume $A$ and $B$ have respective parameters $x_a$ and $x_b$ and a shared pairwise mask $s_{a,b}$. They simply hide their parameters by uploading $y_a = x_a + s_{a,b}$ and $y_b = x_b - s_{a,b}$ (with an appropriate modulus) respectively. The shared mask will be cancelled during aggregation without distorting the parameters. While the shared mask can be conveniently generated using well-known key exchange protocols and pseudo random generators, the main problem is to deal with dropout users who become offline in the middle of the process and make the shared mask not cancellable. To address this problem, Bonawitz et. al. \cite{bonawitz2017practical} proposed a protocol that allows online users to recover offline users' secrets through secret sharing. Without heavy cryptographic primitives, \cite{bonawitz2017practical} supports secure (linear) aggregation for federated learning without accuracy loss as long as the number of malicious users is less than a threshold number $t$.  

\begin{figure}[htbp]
\centerline{\includegraphics[width=0.7\linewidth]{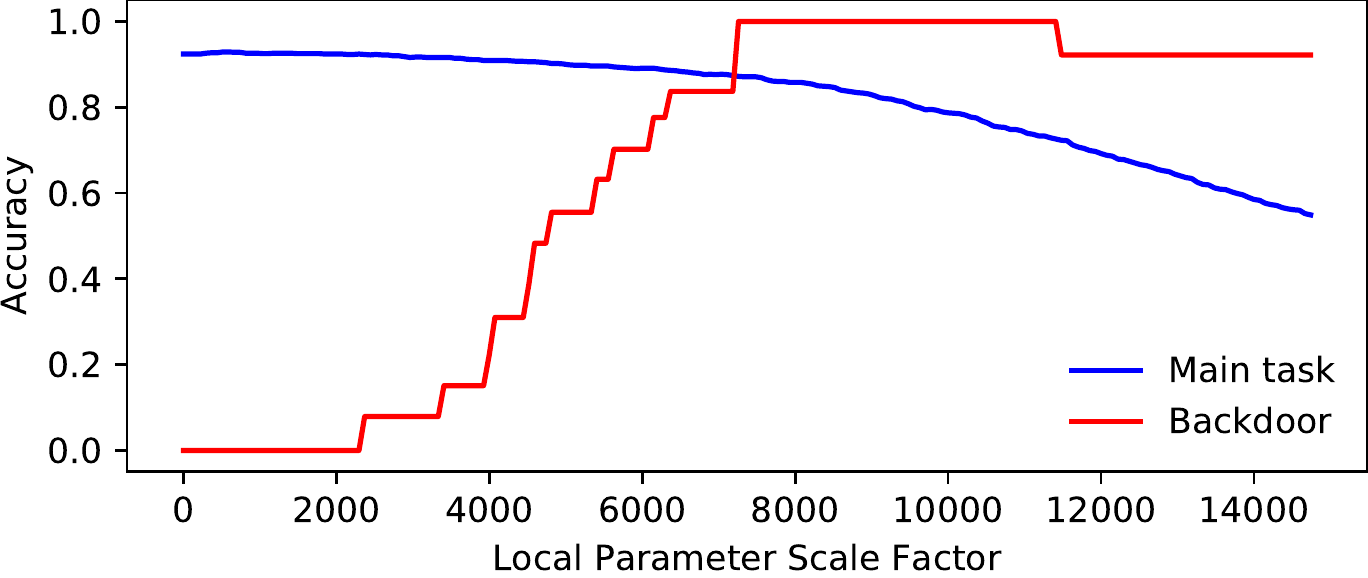}}
\vspace{-2.5mm}
\caption{Backdoor/main model accuracy of \cite{bonawitz2017practical} with one attacker.}\label{f:backdoor}
\vspace{-4.5mm}
\end{figure}

One outstanding problem with secure aggregation is that model poisoning attacks are made easier. This is because local models are no longer revealed to the aggregator, which is required in existing Byzantine attack detection techniques \cite{fung2018mitigating}.  By uploading scaled erroneous parameters, attackers can launch model poisoning attacks through label-flipping \cite{biggio2012poisoning} or model backdoor \cite{bagdasaryan2018backdoor}, with the aim to manipulate the global model at the attacker's will. As shown in Fig. \ref{f:backdoor}, a backdoor can be conveniently inserted to the global model in \cite{bonawitz2017practical} even if one single attacker is presented. To our best knowledge, it still remains a challenge to simultaneously support backdoor detection and secure aggregation in federated learning.


In this paper we design a new protocol, namely SAFELearning, to address backdoor attacks in secure aggregation. We achieve this with two new primitives - \emph{oblivious random grouping (ORG)} and \emph{partial parameter disclosure (PPD)}. ORG partitions users into one-time \emph{random} subgroups at each iteration with both group membership information and subgroup configurations (including the number of subgroups) \emph{oblivious} to users (membership information is also oblivious to the aggregation server). Neither users nor the aggregation server can predict or determine which user is assigned to which subgroup before local parameters are uploaded. This property forces attackers to work independently even if they are willing to collude. By making subgroup configurations oblivious, we further thwart opportunistic attackers who independently manipulates local parameters based on statistical estimation of the distribution of attackers in subgroups. PPD supports secure partial disclosure of aggregated models of subgroups by limiting the privacy leakage to the level that the global model can leak. With ORG and PPD, the aggregation server is able to randomly evaluate subgroup models for anomaly detection without additional privacy leakage.  

As compared to Bonawitz et. al., the computational complexity of SAFELearning is reduced from $O(N^2+mN)$ (at user side) and $O(mN^2)$ (at server side) to $O(N+n^2+m)$ and $O(mN+nN)$, respectively, where $N$ is the total number of users, $n$ ($\ll N$) the number of users in subgroups and $m$ the size of the model. This is attributed to the hierarchical subgroup design in ORG. SAFELearning is provably secure under the simulation-based model. We conducted extensive experiments with ResNet-18 network over CIFAR-10 dataset under well-known backdoor strategies. Experimental results show that SAFELearning is able to completely defeat backdoor attacks (i.e., reducing backdoor accuracy to 0) in most cases when the percentage of attackers is relatively low. Meanwhile, SAFELearning exhibits similar convergence rate as compared to vanilla federated learning. Table \ref{tab:cmp} compares SAFELearning with existing representative secure federated learning techniques.   
\vspace{-1em}
\begin{table}[!htb]
    \caption{Related secure federated learning techniques}\label{tab:cmp}
    \small
    \vspace{-2mm}
    \centering
    \begin{tabular}{ |c|c|c|c|} 
        \hline
        \thead{Method}&\thead{Privacy\\ Preservation}& \thead{Backdoor\\ detection}&\thead{Computation\\ complexity}\\
        \hline
        \makecell{SecureFL\cite{bonawitz2017practical}} &\cellcolor{green!25} Yes & \cellcolor{red!25}No& \cellcolor{red!25}$O(mN^2)$\\
        \hline
        \makecell{ Detection\cite{blanchard2017machine}\cite{fung2018mitigating}}&\cellcolor{red!25}No&\cellcolor{green!25}Yes&\cellcolor{green!25}$O(N^2)$\\
        \hline
        SAFELearning&\cellcolor{green!25}Yes&\cellcolor{green!25}Yes&\cellcolor{green!25}$O(mN+nN)$\\
        \hline
    \end{tabular}
    
    \vspace{-3mm}    
\end{table}  

The main contributions of this paper are as follows: 1) we design a new secure aggregation protocol that simultaneously supports backdoor attack detection and model privacy; to our best knowledge, this work is among the first that can detect model-poisoning attacks on encrypted model parameters for federated learning; 2) the proposed scheme significantly improves system scalability in both computational and communication complexities as compared to state-of-the-art secure aggregation techniques; 3) SAFELearning is provably secure under the simulation-based model. 

This rest of the paper is organized as follows. Section \ref{s:System} presents models and technical preliminaries. An overview to our design is described in Section \ref{s:overview}. Section \ref{s:tree} elaborates our ORG primitive and secure aggregation protocol. Section \ref{s:malicious} explains the PPD design and our backdoor detection mechanism. Section \ref{s:Eval} presents experimental evaluation of our backdoor detection algorithm. Section \ref{s:Related} review related work. We conclude the paper in Section \ref{s:Conclusion}.

\section{Models and Preliminaries}\label{s:System}

\subsection{System Model}
We assume two types of parties in a federated learning system: an aggregation server $S$ and a set of $N$ participating users $U$. The server holds a global model $X_i$ of size $m$ and each user $u\in U$ possesses a private training data set. Users train the global model shared by the server with their private training data at each iteration and upload the local model parameters to the server. The server aggregates local parameters and compute $\sum_{u \in U} x_u$, where $x_u$ (also of size $m$) is the local model parameter trained by $u$ using $X_i$ and his local data. The server returns the latest global model to each user at the end of each iteration. The server communicates with each user through a secure (private and authenticated) channel. A trusted third party authenticates each user and generates a private-public key pair $<s_u^{SK},s_u^{PK}>$ for her before the protocol execution. Each user $u$ is assigned a unique ``logical identity" $Id_u$ in a full order. Users may drop out at any phase of the protocol but at least $t$ users are online for any randomly selected subgroup of $n$ users at any time. In the training phase, we assume the server uses the baseline federated learning algorithm FedSGD\cite{mcmahan2017communication} with the following update rule, where $\eta$ is the learning rate: $X_{i+1} = X_{i} + \frac{\eta}{N}\sum\limits_{x_u\in U}(x_u-X_i)$. 


%
%
%

\subsection{Attack Model}\label{ss:attackmodel}
We consider two types of attackers who have distinct objectives: \emph{type \rom{1} attackers} are interested in learning the value of local parameters to comprise model privacy; and \emph{type \rom{2} attackers} are motivated to poison the global model to generate a model with semantic backdoors. The global model in type \rom{2} attack shall exhibit a good main task accuracy (or ``main accuracy" for short) but also behave in a way at attacker's will on attacker-chosen backdoor inputs.

  \textbf{Type \rom{1} attackers:}
Both the server and malicious users are considered potential type \rom{1} attackers but the server is assumed honest but curious. The server may collude with malicious users to compromise model privacy of benign users. However, the ratio of users it can collude does not exceed $\frac{r}{N}$ for any randomly chosen subgroup of users. We assume the server is against type \rom{2} attackers since it may benefit from an accurate and robust global model. 

  \textbf{Type \rom{2} attackers:}
For type \rom{2} attackers, we follow similar assumptions of \cite{bagdasaryan2018backdoor}, i.e., attackers have full control over one or several users, including their local training data, models, and training settings such as learning rate, batch size and the number of epochs. However, attackers are cannot influence the behavior of benign users. The number of type \rom{2} attackers (malicious or compromised users) are assumed to be much smaller than benign users. Further more, we assume attackers can fully cooperate with each other, including sharing their secret keys, whenever necessary. To be specific, we define the objective of type \rom{2} attackers $U_{a}$ in each iteration is to replace the aggregated global model $X_{i+1}$ with a target model $X_{target}$ \footnote{$X_{target}$ can be a transitional model in continuous attacks.} as shown in (\ref{eq:atk2}). $U_h$ ($U_a$) refers to the set of benign users (attackers) in this iteration.
\begin{equation}\label{eq:atk2}
    X_{i+1}=X_i+\frac{\eta}{N}\sum\limits_{x_u\in U_{h}}(x_u-X_i)+\frac{\eta}{N}\sum\limits_{x_a\in U_{a}}(x_a-X_i)\approx X_{target}
\end{equation}

\subsection{Attacker Strategies}\label{ss:attackstra}
To yield a target model $X_{target}$, each attacker shall construct local parameters $x_a$ approximately base on the following:

\begin{align}\label{eq:atk4}
    \begin{split}
        \overline{x_a-X_i} &= \frac{N}{\eta |U_a|}(X_{target}-X_i)-\frac{1}{|U_a|}\sum\limits_{x_u\in U_{h}}(x_u-X_i)\\
     &\approx\frac{N}{\eta |U_a|}(X_{target} - X_i)
    \end{split}
\end{align}

To make the attack more effective and stealthy, the following possible strategies might be adopted by attackers :

  \textbf{Sybil attacks:} In order to reduce the scaling factor $\gamma = \frac{N}{\eta |U_a|}$ and make the attack stealthy,  attackers tend to deploy as many adversary participants as possible, e.g., by Sybil attacks \cite{douceur2002sybil} to increase $|U_a|$. In this work, we assume the trusted party authenticates each user when issuing public/private key pairs to thwart sybil attacks. 
    
 \textbf{Adaptive Attacks:} Strategic attackers can launch adaptive attacks \cite{bhagoji2019analyzing,bagdasaryan2018backdoor} by including the distance $Distance(X_{target},X_i)$ between $X_{target}$ and $X_i$ (either geometric distance or cosine distance of gradients) in the loss function while training $X_{target}$.  The purpose is to reduce the term $X_{target}-X_i$ in (\ref{eq:atk4}) to make attack more imperceptible. However, this term cannot be arbitrarily optimized if the attackers' objective is different from the main task of the model.
    
 \begin{figure}[htbp]
\centerline{\includegraphics[width=0.8\linewidth]{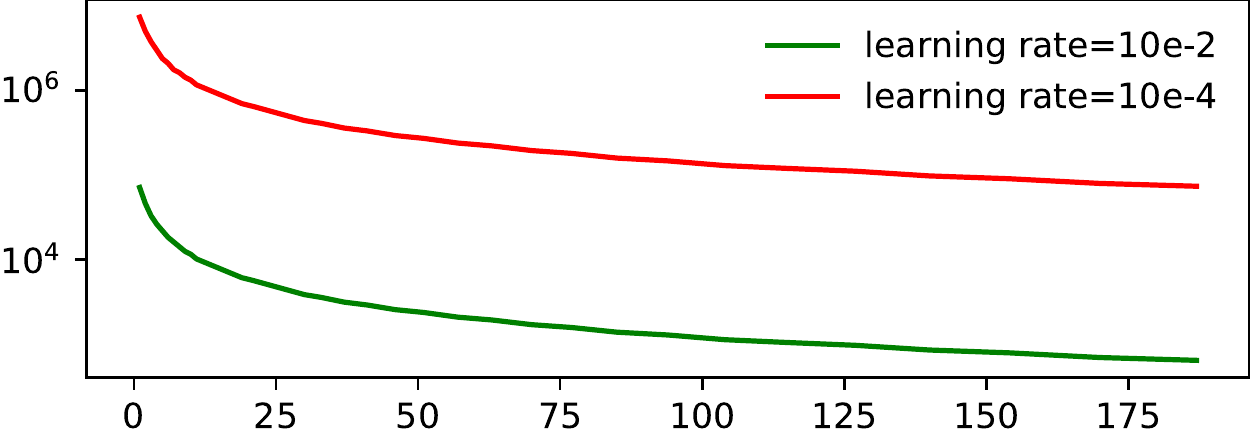}}
\vspace{-3mm}
\caption{Continuous attacks - scale factor vs. number of iterations}\label{fig:AtkS}
\vspace{-4mm}
\end{figure}
    
   \textbf{Continuous Attacks:} To make the attack stealthier, attackers may instead choose to continuously perform the attack through multiple iterations, expecting a relatively smaller scale factor at each iteration. However, our preliminary experiments in Fig. \ref{fig:AtkS} discovers that a minimum parameter scale factor is required for successful backdoor attacks. The scale factor is proportional to the inverse of learning rate which is controlled by the server. Previous research \cite{xie2019dba}\cite{bagdasaryan2018backdoor} also discovered that rationale attackers would launch the attack toward the convergence time. This is because if launched in early stages the backdoor tends to be "forgotten" and the main task will be jeopardized due to the disturbance of the attack.
   
Therefore, in this paper we assume attackers need to scale up their model to a certain minimum level even with continuous attacks. We also assume the attack will not be launched at the very beginning of training but rather in the middle of it to reduce the disturbance on the main task. We consider all the above attack strategies and their combinations except for unlimited Sybil attacks. Following existing research \cite{bonawitz2017practical}, we assume the attackers only count for a small portion of the entire population of users (i.e., $|U_a| \ll N$).

\subsection{Secure Aggregation with Pairwise Masking}\label{s:Secure}
Recently, pairwise additive masking \cite{bonawitz2017practical, AC11, GX15} has been utilized as an efficient cryptographic primitive for secure aggregation in federated learning even for complex deep networks.  As this paper utilizes pairwise masking for secure aggregation, we provide overview of a recent secure aggregation scheme \cite{bonawitz2017practical} as follows. 


Let $x_u$ denote the an $m$-dimensional vector of parameters that user $u\in U$ generates locally, where $U$ is the set of all users. Assume a total order of users and each user $u$ is assigned a private-public key pair $(s_u^{SK},s_u^{PK})$. Each pair of users $(u,v)$, $u<v$, can agree on a random common seed $s_{u,v}$ using Diffie-Hellman key agreement\cite{diffie1976new}. With the seed, a common mask vector ${PRG}(s_{u,v})$ can be computed by $u$ and $v$ using a pseudo-random generator ($PRG$) (e.g., a hash function). When $u$ obfuscates her parameter vector $x_u$ by adding the mask vector and $v$ subtracting it, the mask vector will be canceled when server aggregates the obfuscated parameter vectors without reveal their actual values. Specifically, each user $u$ obfuscates her parameter $x_u$ as following:
\begin{align*}
y_u = x_u&+ \sum\limits_{\forall v\in U:u<v}PRG(s_{u,v})\\
&- \sum\limits_{\forall v\in U:u>v}PRG(s_{v,u})\quad (mod\quad R)
\end{align*}
and sends $y_u$ to the server. Then the server computes:
\begin{align*}
z &= \sum\limits_{u\in U}\bigg(x_u+ \sum\limits_{v\in U:u<v}PRG(s_{u,v})- \sum\limits_{v\in U:u>v}PRG(s_{v,u})\bigg)\\
  &= \sum\limits_{u\in U}x_u\quad (mod\quad R)
\end{align*}

To address user dropouts, each user $u$ creates $N$ shares of her secret $s_u^{SK}$ using Shamir's $t$-out-of-$N$ secret sharing scheme and sends the shares  to the rest of users. Additionally, each user $u$ generates another random seed $b_u$ which is mainly to prevent the aggregation server from learning her parameter vectors in case she is delayed but her secret has been recovered by other users before she becomes online and sends out $y_u$. Random shares of $b_u$ are also generated and sent to other users. Each user $u$ obfuscates the parameter vector $x_u$ using a mask ${PRG}(b_u)$ in addition to the pairwise mask vector:
\begin{align*}
y_u = x_u&+PRG(b_u) + \sum\limits_{\forall v\in U:u<v}PRG(s_{u,v})\\&- \sum\limits_{\forall v\in U:u>v}PRG(s_{v,u})\quad (mod\quad R)
\end{align*}

In the unmask round, for each dropped user $v$, online users reveal the shares of $s_v^{SK}$; for each online users $u$, other online users reveal the shares of $b_u$. Then the server will be able to compute $PRG(s_{v,u})$ and $PRG(b_u)$ for any online user $u$ and cancel it out from the sum $z$ to get the aggregate model of online users. Note that an honest user $u$ never reveals either shares of $s_{j,v}$ or $b_v$ for any user $v$ before the unmask round.

The scheme saliently protects confidentiality of local parameters efficiently while taking into account user dropouts in practical distributed systems. However, the solution also makes it convenient for model poisoning attacks (e.g., backdoor attack). As pointed out in a recent work by Bagdasaryan et. al. \cite{bagdasaryan2018backdoor}, even a single malicious user is able to manipulate the global model through model replacement attack. This is possible because secure aggregation fully encrypts the users' local model, which allows the attacker to submit any erroneous parameters. As the aggregation server does not necessarily have access to validation datasets, such attack is difficult to detect by simple model validation. 

\section{Overview of SAFELearning Design}\label{s:overview}
\begin{figure*}[htbp]
\centerline{\includegraphics[width=\textwidth]{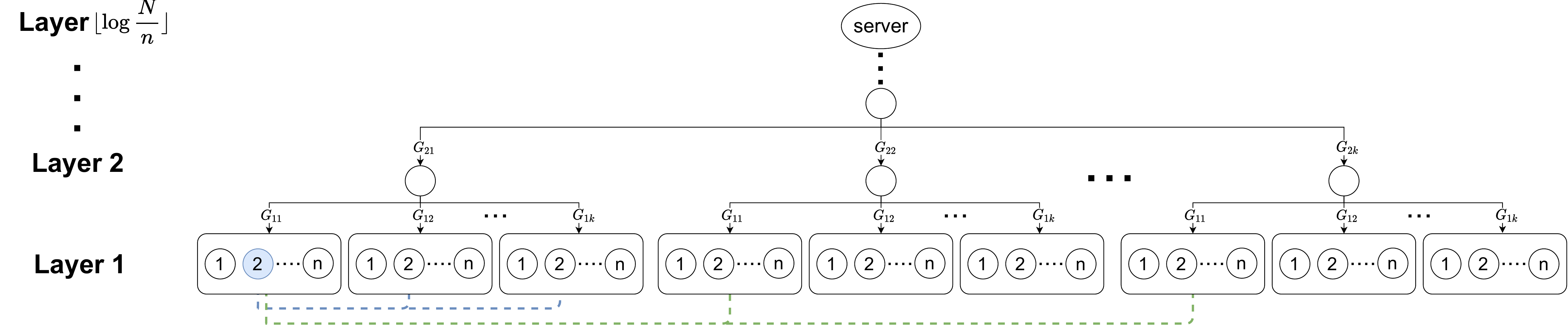}}
\vspace{-2mm}
\caption{The random tree structure (blue and green dotted lines indicate subgroups at levels 2 and 3 respectively). }
\label{Tree_1}
\vspace{-5mm}
\end{figure*}
Achieving both model privacy and backdoor detectability seems challenging because of the conflicting natures of the two objectives. To address this issue, in SAFELearning we organize users into subgroups with a hierarchical k-ary tree structure as shown in Fig. \ref{Tree_1}. At the leaves are equal-sized subgroups of $n$ users. At the aggregation server, the models of each subgroups are first aggregated; the aggregated model of the subgroups are further aggregated at the next level of subgroups; the process repeats recursively toward the root of the tree. It is trivial to show that the aggregated global model remain the same as exiting federated learning algorithms. With the tree structure, users in the same subgroup pairwise ``mask" each other during the secure aggregation process. To protect privacy of the aggregated model of each subgroup, a pairwise mask is also generated for each subgroup at internal layers of the tree as shown in dashed line in Fig. \ref{Tree_1}. Similarly, user secrets (i.e., $s_u^{SK}$ and $b_u$ for user $u$ as discussed in Section \ref{s:Secure}) can be securely shared within subgroups. Intuitively, secure aggregation with the tree structure provides similar level of protection to model privacy of users as in \cite{bonawitz2017practical}, and the user dropouts can be handled similarly as well. As secret sharing is within subgroups, we directly enjoy the benefit of reduced complexity, i.e., from $O(N^2)$ in \cite{bonawitz2017practical} to $O(n^2)$, because of the hierarchical group design.  

However, strategic type \rom{1} attackers can comprise local model privacy by deploying an overwhelming number of malicious users in a target subgroup. This is possible when $n \ll N$ and $|U_a| \not \ll n$. To prevent such attacks, we shall not allow either the server or any user to determine which users belong to which subgroups. Specifically, the assignment of users to subgroups shall be \emph{randomized} so that nobody can assign herself or others to a target subgroup with a non-negligible probability more than random assignment. 

The randomized subgroup assignment also provides the opportunity for detecting type \rom{2} attackers, whose purpose is to manipulate the aggregated global model, e.g., to insert backdoors. Specifically, model poisoning attackers need to amplify their local parameters dramatically in order to influence the global model as discussed in Section \ref{ss:attackstra}, no matter in ``one-shot" last-round attacks or continuous attacks. If each attacker were to work independently, the magnitude of aggregated models at the subgroups will differ significantly from each other unless each subgroup has exactly the same number of attackers, the chance of which however is very low due to the randomness of the tree-based subgroup assignment even if the number of attackers equals the number of subgroups. The aggregation server can make the chance even lower by randomly changing subgroup configurations over iterations. Considering collaborative attackers, however, such randomness alone is not enough for backdoor attack detection (or model poisoning detection in general). This is because collaborative attackers can intentionally adjust the scales of their local parameters to make the distribution of aggregated subgroup models uniform unless subgroups (at leaf layer) outnumbers attackers. To defeat such collaboration, we need to make the attackers  \emph{oblivious} to each other's subgroup membership information. This means that the attackers shall not know whether or not a given user/attacker belongs to which subgroup. We call such subgroup assignment as \emph{oblivious random grouping (ORG)}.   

With the tree-based ORG, we can detect model poisoning attacks by evaluating aggregated models of the subgroups (i.e., those at the leaf layer of the tree). Specifically, subgroups with more attackers will have much higher magnitude in their aggregated models because of the scale-up of attackers' parameters unless each subgroup has exact number of attackers, the chance of which is extremely low in large-scale systems as shown in Section \ref{ss:OSA}. However, directly revealing aggregated models of subgroups may also lead to some privacy leakage depending on the number of users in the subgroup. To address this issue, SAFELearning only reveals partial higher bits of aggregated subgroup models to the extend that the privacy leakage is no more than what is disclosed by the global model (which is public to all users anyway). Such \emph{partial parameter disclosure (PPD)} allows us to compare aggregated models of subgroups and even to conduct some statistical analysis with model privacy preserved.     

The high-level workflow of our protocol is as following: (a) first, the users and server work together to generate the random tree with our tree generation sub-protocol; at the end of this step, the tree structure and the full orders of secret sharing and pair-wise masking are determined; (b) next, each user shares his secret keys $s_u^{SK}$ and $b_u$ to the users in the same subgroup, obliviously masks his input $x_u$ according to the tree structure (without knowing the subgroup membership information), and sends the encrypted input to the server. After having collected the input and secret shares from the users, (c) the server compares the partial information of the aggregated model from each subgroup to detect abnormal subgroup(s). It computes the global model and returns it to users but without using parameters of abnormal subgroups. ORG is implemented through steps (a) and (b), and PPD is realized in steps (b) and (c). Next two sections elaborate our design of tree-based secure aggregation and poisoning attack detection.


\section{Tree-Based ORG and Secure Aggregation}\label{s:tree}

\subsection{Tree Structure Definition}

We define two independent random tree structures $T_{share}$ and $T_{masking}$, both of which share the same structure of Fig. \ref{Tree_1} but with independent subgroup membership assignments. $T_{share}$ is for secret sharing of secret keys $s_u^{SK}$ and $b_u$, and $T_{masking}$ is for pairwise masking of model parameters. Decoupling the two trees is because of the different security requirements of secret sharing and pairwise masking. For secret sharing, tree $T_{share}$ divides users into small subgroups (rounded rectangle shows in Fig. \ref{Tree_1}). The secret sharing will perform between users inside the same subgroup at the leaf layer of the tree. 

Secure aggregation is performed as in Section \ref{s:Secure} but using pairwise masks within subgroups. Each user applies two types of masking to her parameters - \emph{intra-group masking} and \emph{inter-group masking}. Intra-group masking is to protect individual local models using pairwise masks of users within the same subgroup at the leaf level of the tree. The aggregated model obtained at this level is called \emph{subgroup aggregated model}. Inter-group masking is to protect the subgroup aggregated models, and the pairwise masks are generated between peers at higher non-leaf layer subgroups. As shown in  Fig. \ref{Tree_1}, we logically form groups (represented as $G_{ij}$, $i$ is the layer number and $j$ means this group is the $j$-th child node of its parent) over $k$ subgroups at each non-leaf layer of the tree. This process repeats recursively toward the root. Specifically, for each user $u$ a pairwise masking peer $v$ should comply with following rules:

(1) Ancestors of $u$ and $v$ immediately under their least common ancestor (LCA) shall be within $\kappa$ immediate neighborhood based on the ancestors' total oder at that layer, where $\kappa$ is a system parameter.

(2) The positions of $u$ and $v$ shall be the same by their total orders in their respectively sub-trees, so are the positions of their ancestors below the immediate children of their LCA. 

(3) Only two peers are needed at each layer for each node.


Based on these rules, when $\kappa = 1$ user 2 in $G_{11}$ of $G_{21}$ in Fig. \ref{Tree_1} has the following peers:  users 1 and 3 in the same subgroup; user 2 in $G_{12}$ of $G_{21}$ and user 2 in $G_{1k}$ of $G_{21}$ (because $G_{11}$ is the neighbor (mod k) of $G_{12}$ and $G_{1k}$ at layer 1); user 2 in $G_{11}$ of $G_{22}$ and user 2 in $G_{11}$ of $G_{2k}$ (because $G_{21}$ is the neighbor (mod k) of $G_{22}$ and $G_{2k}$ at layer 2); so on and so forth. This is illustrated by the dotted line in Fig. \ref{Tree_1}. Therefore, our inter-group masking have following properties: first, the number of pairwise masking operations for each user is $2\log\frac{N}{n}$, twice of the tree height; second, those pairwise masks cannot be cancelled until the server aggregates all the subgroups' aggregated models at that layer of the tree. Let $G_{u}^{s}$ be the intra-group masking peers of user $u$ (i.e., pairwise peers in the same subgroups as $u$) and $G_u^{p}$ inter-group masking peers. The masking equation for user $u$ can be written as 
\begin{align}\label{eq:mask}
\begin{split}
y_u = x_u&+PRG(b_u) + \sum\limits_{\forall v\in \{G_{u}^{s},G_{u}^{p}\}:u<v}PRG(s_{u,v})\\&- \sum\limits_{\forall v\in \{G_{u}^{s},G_{u}^{p}\}:u>v}PRG(s_{v,u})\quad (mod\quad R)
\end{split}
\end{align}

The server is able to recover a dropped user $u$'s pairwise masks of any type by secure recovery of his private key $s_u^{SK}$ if at least $t$ honest users in the dropped user's secret sharing subgroup survived.

\subsection{Tree-Based Random Subgroup Generation}\label{ss:RSG}
For secure ORG, the first requirement is the randomness of the subgroup membership assignment. As users are grouped in the total order by their identities, the randomness can be assured if user identities are randomly generated, i.e., they are random and not solely determined by either the user herself or the server. Specifically, user identity $Id_u$ for user $u$ is generated as following\footnote{It works similarly for $T_{share}$ and $T_{mask}$. In our next description, we take $T_{share}$ as example.}: $Id_u = HASH(R_s||c^{PK}_{u}||R_u)$, where $c^{PK}_{u}$(or $s^{PK}_{u}$ if it were the pairwise mask tree $T_{mask}$) is the public key used in Diffie-Hellman key agreement and is used one-time for each iteration. Random numbers $R_s$ and $R_u$ are generated by the server and user $u$ respectively. Because of the randomness of hash function, $Id_u$ is randomly distributed and not predictable to both users and the server.





For random subgroup assignment, however, the order of the disclosure of the tree structure $T$ (generated by the server) and that of $Id_u$ (jointly produced by the server and user $u$) is important. In particular, if $Id_u$ is disclosed before $T$, the server might be able to intentionally group certain users in a subgroup by adjusting the tree structure. On the other hand, if $T$ is disclosed before $Id_u$, malicious users could attempt to group themselves together by manipulating their identities (e.g., via finding special hash results). Both could lead to model privacy disclosure of victim users. To defeat such potential attacks, we design a commitment protocol, as shown in Fig. \ref{f:SubP}, with which $T$ and $Id_u$ are committed before disclosure. Specifically, the server first broadcasts its commitment of a random number $R_s$. Users then send their public key $c^{PK}_u$ (or $s^{PK}_{u}$ if it were pairwise mask tree $T_{mask}$) and commitment of $R_u$ to the server. After collecting enough users, the server will decide the tree structure (degree and layers) $T$ base on the number of users $N$. At this point, $Id_u$ has been uniquely determined but yet disclosed; neither users nor the server can compute or predict it. And the tree structure is determined by the server independently to user identities. The server then broadcasts the commitment of $T$, $R_s$ and a list of commitments to $R_u$'s. On receiving the broadcast message from the server, user discloses $R_u$ to the server, allowing the server to compute the user's $Id_u$ and make subgroup assignment. Users can verify the correctness of the protocol by requesting the server to broadcast $T$ and the list of $c_u^{PK}$(or list of $s_u^{PK}$) and $R_u$ after local parameters have been sent out. 

\begin{figure}[htbp]
\vspace{-3mm}
\centerline{\includegraphics[width=0.35\textwidth]{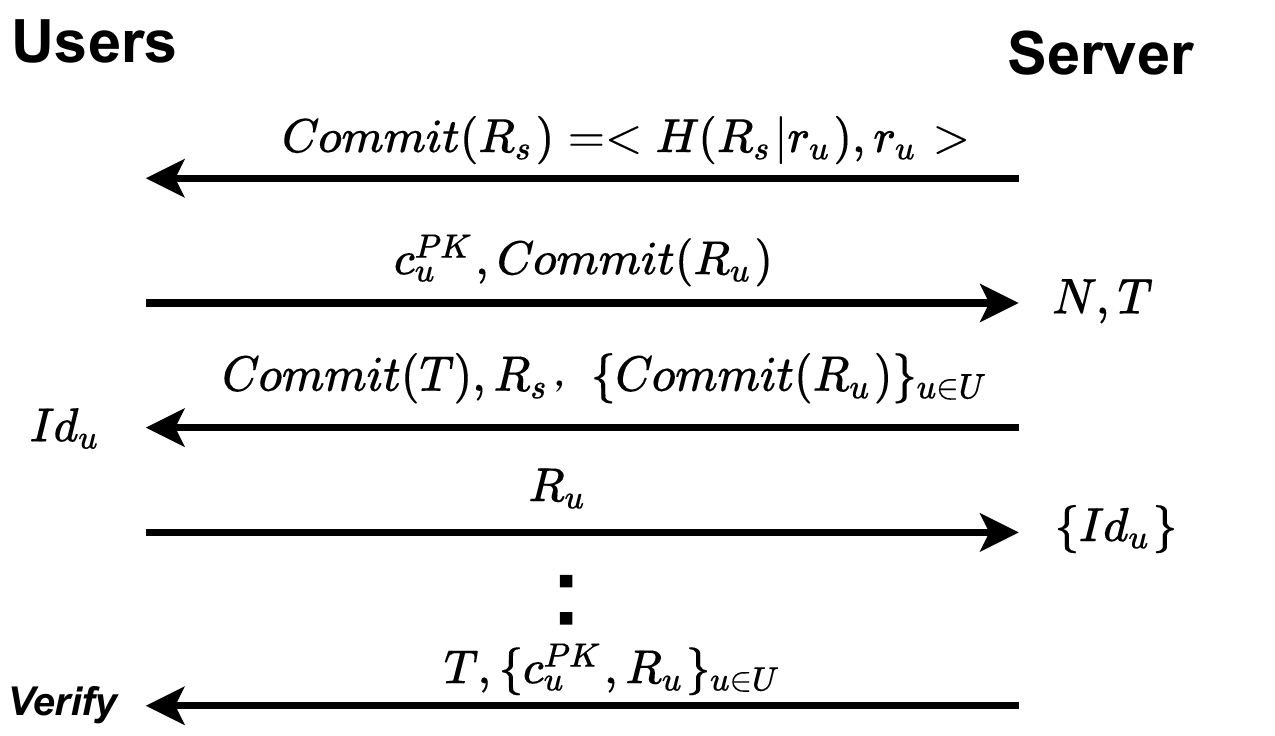}}
\vspace{-1mm}
\caption{Tree structure generation protocol.}
\label{f:SubP}
\vspace{-3mm}
\end{figure}

The protocol has considered misbehaving server and users. However, there is still risks when the server colludes with malicious users. For example, they can pre-compute $Id_u$ by exhaustively testing different $R_s$ and $R_u$ to obtain special identities, e.g., those with leading zero(s). As benign users are less likely to have such special identities due to randomness of the hash function, the malicious users will be assigned to the same subgroup and dominate that group. To defeat such attacks, we generate the final identity of user $u$ as $HASH(\smashoperator[r]{\sum_{\forall v\in G:v\neq u}} Id_v)$ instead of using $Id_u$ directly, where $\sum$ means XOR. By this, the randomness is determined by all but the user himself. With random and trustworthy full order identities of all users and an independently generated tree structure, we can achieve random subgroup assignment as shown in Fig. \ref{Tree_1}. 


\subsection{Oblivious Secure Aggregation}\label{ss:OSA}

As discussed, users shall be oblivious to subgroup membership information for secure ORG. Otherwise, malicious users are able to coordinate and manipulate the distribution of the subgroup aggregated model parameters to bypass anomaly detection. In particular, if attackers (malicious users) outnumbers the subgroups (at the leaf layer), they can coordinate and strategically adjust local parameters to make the distribution of subgroup aggregated models uniform. If the subgroups outnumbers attackers, however, the distribution is doomed imbalance unless attackers give up the attack. Without coordination, the chance that each subgroup contains exact the same number of attackers is very low. For example, if there are $x$ attackers and $x$ subgroups, the probability that each subgroups have exactly one attacker is $\frac{x!}{x^x}$.

However, if users are oblivious to subgroup membership, they are not able to identify their peers and generate pairwise masks to encrypt local parameters. To solve this problem, we let the server (who is against model poisoning attacks for its own benefits) directly send each user the list of public keys of all the users who are her pairwise mask peers. However, directly sending original public keys may allow malicious users to recognize each other and know their group membership information. To address this problem, we design a randomized D-H key exchange protocol wherein the server randomizes each user's public key before sending it out. With this randomized public key, two malicious users are not able to tell whether or not they belongs to the same subgroup unless they are pairwise peers. Specifically, our construction is as follows.   

\textbf{Randomized D-H key exchange.} Assume the server is to coordinate the exchange of public keys between users $u$ and $v$, with their respective public keys $s_u^{PK} = g^{s_u^{SK}}$ and $s_v^{PK} = g^{s_v^{SK}}$. To prevent them from recognizing each other's public key, the server ``randomizes" their public keys before sending out. Specifically, it first produces a random number $r_{u,v}$, and then sends a randomized public key $s_u^{(PK,v)} = (s_u^{PK})^{r_{u,v}}$ to user $v$ and $s_v^{(PK,u)} = (s_v^{PK})^{r_{u,v}}$ to user $u$. After key exchange like D-H, the shared key $s_{u,v}$ will become following form:
\begin{equation*}
    s_{u,v}=(s_u^{(PK,v)})^{s_v^{SK}} = (s_v^{(PK,u)})^{s_u^{SK}} = g^{s_u^{SK}*s_v^{SK}*r_{u,v}}
\end{equation*}

Please note that pairwise peers are still able to verify that they are in the same subgroup by comparing the shared key they computed. However, the purpose of our randomized D-H key exchange is to thwart users who are in the same subgroup but not peers from knowing the fact that they are in the same subgroup. This is achieved because of the unique random number $r_{u,v}$ for each pair of peers. Without the randomization, however, two attackers will easily know that they in the same group if they receive common public keys of benign user(s).

%
\begin{figure}[htbp]
\centering
\begin{subfigure}{.5\linewidth}
  \centering
  \includegraphics[width=.7\linewidth]{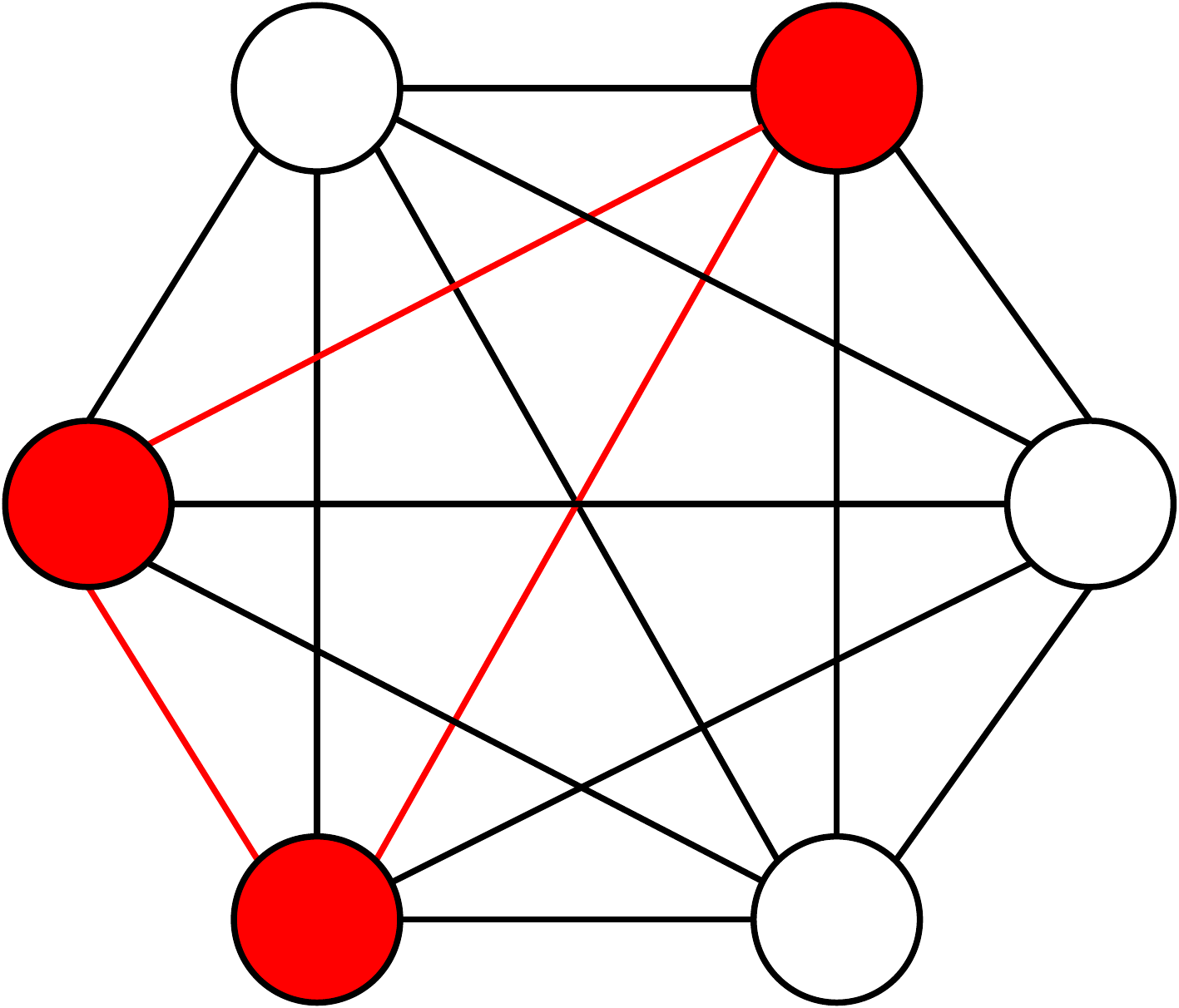}
  \caption{Complete graph topology}
  \label{fig:sub1}
\end{subfigure}%
\begin{subfigure}{.5\linewidth}
  \centering
  \includegraphics[width=.7\linewidth]{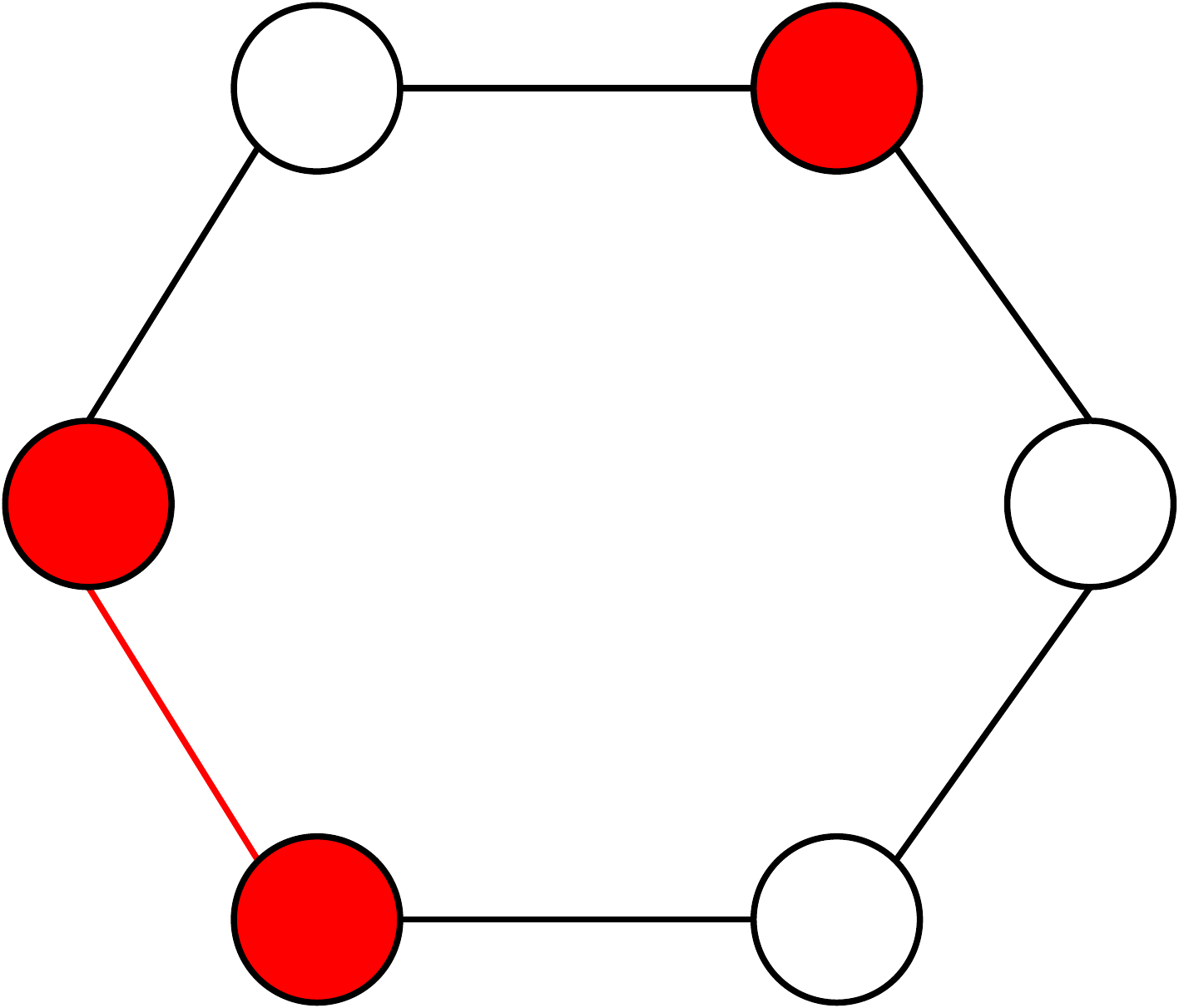}
  \caption{Circular topology}
  \label{fig:sub2}
\end{subfigure}
\vspace{-2mm}
\caption{Graph representation of users inside a subgroup. Red nodes represent attackers. Edges denote pairwise masking relationship.}
\label{fig:graph}
\vspace{-3mm}
\end{figure}
\textbf{Circular topology for pairwise masking relationship} Consider users inside a subgroup as vertices of an undirected graph, the edges of which represent pairwise masking peer relationship as shown Fig. \ref{fig:graph}. If users choose all the other users in the subgroup as peers, the graph is complete. In this case attackers will know each other's membership information because each pair of peers share the same seed $s_{u,v}$. To avoid such situation, we use a circular graph as shown in Fig. \ref{fig:graph} (b), in which each user $u$ only pairs with $2\kappa$ users - her $\kappa$ immediate previous and next neighbors based on the total order $Id_u$. For successful attack, attackers need to know the membership information of all of them. This is possible only when the attackers inside the same subgroup form a chain. Assume there are a total of $x$ attackers in a subgroup, the chance is at most $\frac{k\kappa}{n^{k-1}}$, where $k$ is the number of attackers in same subgroup and $n$ is the size of the subgroup.

\section{Backdoor Attack Detection}\label{s:malicious}

In machine learning, one approach to detect model poisoning or Byzantine attacks is by statistically analyzing the magnitude of the model parameter vectors \cite{Shen2016A,blanchard2017machine}. As discussed in Section \ref{ss:attackstra}, attackers need to amplify model parameters by a factor of $\gamma = \frac{N}{\eta |U_a|}$ on average to successfully launch model poisoning attacks. With random subgroup assignment (Section \ref{ss:RSG}) and oblivious secure aggregation (Section \ref{ss:OSA}), users are partitioned into random subgroups at each iteration and perform pairwise masking without knowing others' group information. Moreover, the random tree structure (including the subgroup configuration) is not revealed before local parameters of all users have been uploaded (Section \ref{ss:RSG}). Attackers have to work independently and it is difficult for any of them to control or even predict the distribution of subgroup aggregated models. As a result, the subgroup aggregated parameters will be very sensitive to even slight differences in the number of attackers presented.  This provides the opportunity to detect backdoor attacks by comparing the subgroup aggregated models. 



In our secure aggregation scheme, local parameters are masked by both intra-group masks and inter-group masks. After aggregation at leaf-layer subgroups, the subgroup aggregated models are only protected by inter-group masks. By changing inter-group masks from full-bit masks to partial masks, i.e., by revealing few higher bits, we can observe partial information of subgroup aggregated models and perform anomaly detection. Because attackers need to significant scale parameters, non-zero bits may present in higher bits for attackers' subgroups. However, disclosure of these higher bits is acceptable only if it does not lead to extra privacy leakage as compared to what the global model discloses. To this end we first derive the number of bits that can be disclosed by the partial parameter disclosure (PPD) mechanism and then present our backdoor detection algorithm.

\subsection{Privacy-Preserving PPD}
\textbf{Bit format of model parameters and masking:}  Assume each element of $x_u$ is in range $[0,R_U]$ and can be stored in fix-point format\footnote{Previous research \cite{gupta2015deep} discovered that using $14$ bits fix point data only have $0.05\%$ accuracy decreasing compare to 32 bits float point data in MNIST and CIFAR-10 dataset.}, and the high bit segment $H(x_u)$ in range $[R_H,R_U]$ can be disclosed. The extracted high bits of the aggregated model of subgroup $G_{1i}$  is $\{\sum_{x_u\in G_{1i}}H(x_u)\}_{i\in N,i\le \frac{N}{n}}$. $G_{1i}$ is $i$-th subgroup at the leaf layer (i.e., layer $1$ as shown in Fig. \ref{Tree_1}). To support higher bits disclosure, we can adjust the pairwise masking equation for user $u$ as follows:
\begin{align}
&y_u = x_u+PRG(b_u)+ 
    \smashoperator[r]{\sum_{\forall v\in G_u^s:u<v}}PRG(s_{u,v})-\nonumber\\ 
    &\smashoperator[r]{\sum_{\forall v\in G_u^s:u>v}}PRG(s_{v,u})+
    \smashoperator[r]{\sum\limits_{\forall v\in G^p_u:u<v}}(\Lambda_{R_H}\land PRG(s_{u,v}))-\nonumber\\
    &\smashoperator[r]{\sum\limits_{\forall v\in G^p_u:u>v}}(\Lambda_{R_H}\land PRG(s_{v,u}))\quad (mod\quad R)\label{eq:mask}
\end{align}
where $\Lambda_{R_H}$ is a vector of binary masks of length $m$ and $G^p$ is the set of inter-group peers. The bits at positions in $[log_2R_H,log_2R_U]$ of $\Lambda_{R_H}$ are set as $0$'s, and the rest bits are $1$. In this way, higher bits starting from $log_2R_H$ are not masked. 

To estimate the privacy disclosure by revealing the higher bits of the subgroup aggregated models, we analyze the posterior probability $P(|X_i-y|<\varepsilon|Y)$, where $X_i\in R^m$ is the local parameter vector, $y$ is the aggregation result and $Y$ is its range. We compare the probability bound of aggregating all $N$ users but revealing all the bits (i.e., the disclosure of global model) with the bound of aggregating $n$ users but revealing only bits in range $[R_H,R_U]$ (i.e., disclosure of higher $log_2R_U - log_2R_H$ bits of subgroup aggregated models). We prove that when $R_H=2(1-\sqrt{\frac{n-1}{n}})\varepsilon$, the expected bounds will be the same in the two situations.

\begin{theorem}\label{TH:1}
Let $\{X_i\}_{i\in [1,N]}$ be $N$ samples from some arbitrary distribution with mean vector  $\overline{\mathbb{E}}(X)$ and variance $\overline{\sigma}_N^2$. The probability bound is:
\begin{align}\label{eq:Chev}
    P(||X_i-y||\geq \varepsilon|y=\overline{\mathbb{E}}(X))\leq \frac{\overline{\sigma}_N^2}{\varepsilon^2},\varepsilon>0
\end{align}
\end{theorem}

\begin{proof}
This can be proved by Chebyshev Inequality directly.
\end{proof}

{Theorem \ref{TH:1}} shows that if the server  aggregates at root node, which is equivalent to $Y=y=\overline{\mathbb{E}}(X)$, the difference between local vector $X_i$ and aggregated result $y$ is limited by the variance of $X$.

\begin{theorem}\label{TH:2}
Let $\{X'_i\}_{i\in [1,n]}$ be $n$ samples randomly selected from $\{X_i\}_{i\in [1,N]}$ in {Theorem \ref{TH:1}} and $Y=y\in [\overline{\mathbb{E}}(X')-R_H/2,\overline{ \mathbb{E}}(X')+R_H/2]$, where $\overline{ \mathbb{E}}(X')$ is the mean vector of $\{X'_i\}_{i\in [1,n]}$. If $R_H=2(1-\sqrt{\frac{n-1}{n}})\varepsilon$, the expected bound of $P(|X_i-y|<\sigma|Y)$ is $\frac{\overline{\sigma}_N^2}{\varepsilon^2}$.
\end{theorem}

\begin{proof} Please see Appendix \ref{app:theo2} for detailed proof. 
\end{proof}

{Theorem \ref{TH:2}} indicates that information leakage can be reduced by increasing the subgroup size $n$ or decreasing the number of  bits revealed (i.e., increasing $R_H$). Therefore, the the server can adjust parameters $n$ and $R_H$ to minimize privacy discourse risk. In particular, when $R_H=2(1-\sqrt{\frac{n-1}{n}})\varepsilon$, the privacy leaked by disclosing the higher $log_2R_U - log_2R_H$ bits is the same as what is disclosed by the global model.

\subsection{Backdoor Attack Detection}\label{ss:BAD}

With the partially revealed subgroup aggregated parameters, we design a new backdoor attack detection algorithm (as shown in Algorithm \ref{Alg:1}) to detect and withstand backdoor attacks. At each iteration $t$, we first compute the Euclidean distance $d_i$ between the global model $H(X_t)$ and the partially revealed aggregated model $H(\overline{x_i})$ for each subgroup $i$. Then the server computes the standard deviation $Std(D)_t$ of $\{d_i\}$ for iteration $t$.  An adaptive threshold $\xi_t = \frac{\rho}{T}\sum_{j=1}^{T}Std(D)_{t-j}$ is also computed, which is the product of the average of standard deviations in last $T$ iterations and a threshold expansion factor $\rho$. If the standard deviation $Std(D)_t$ is above the threshold $\xi_t$, the server replaces the subgroup with the largest distance to the global model $H(X_t)$ with $H(X_t)$ (i.e., the partially revealed higher bits), and recomputes the standard deviation of the updated $\{d_i\}$. This repeats until the standard deviation of $\{d_i\}$ is below the threshold $\xi_t$.

To better understand how our method works, we differentiate three attack scenarios: \textsl{all benign subgroups}, \textsl{mixed subgroups} and \textsl{all malicious subgroups}. When no subgroup is malicious, our method can avoid false positive detection because our adaptive threshold is proportional to the average of standard deviation of $\{D\}$ in the last few rounds, and the threshold expansion factor $\rho$ will relax the threshold to accept normal deviations introduced by benign users. When benign subgroups and malicious subgroups are both presented, which usually happens when the number of attackers is less than or close to the number of subgroups, our method will replace the aggregated models of malicious subgroups, which usually have larger deviations from the global model as discussed in Section \ref{ss:attackmodel}, with the current global model. When all subgroups are malicious, which usually happens when the attackers outnumber the subgroups, the subgroups with more attackers will be replaced with the global model first. The remaining malicious subgroups will have large deviations from the global model and result in a large standard deviation $Std(D)_t$ and will be subsequently replaced, which means this iteration of training is abandoned.

 \begin{figure}[htbp]
 \vspace{-1.5mm}
\centerline{\includegraphics[width=0.8\linewidth]{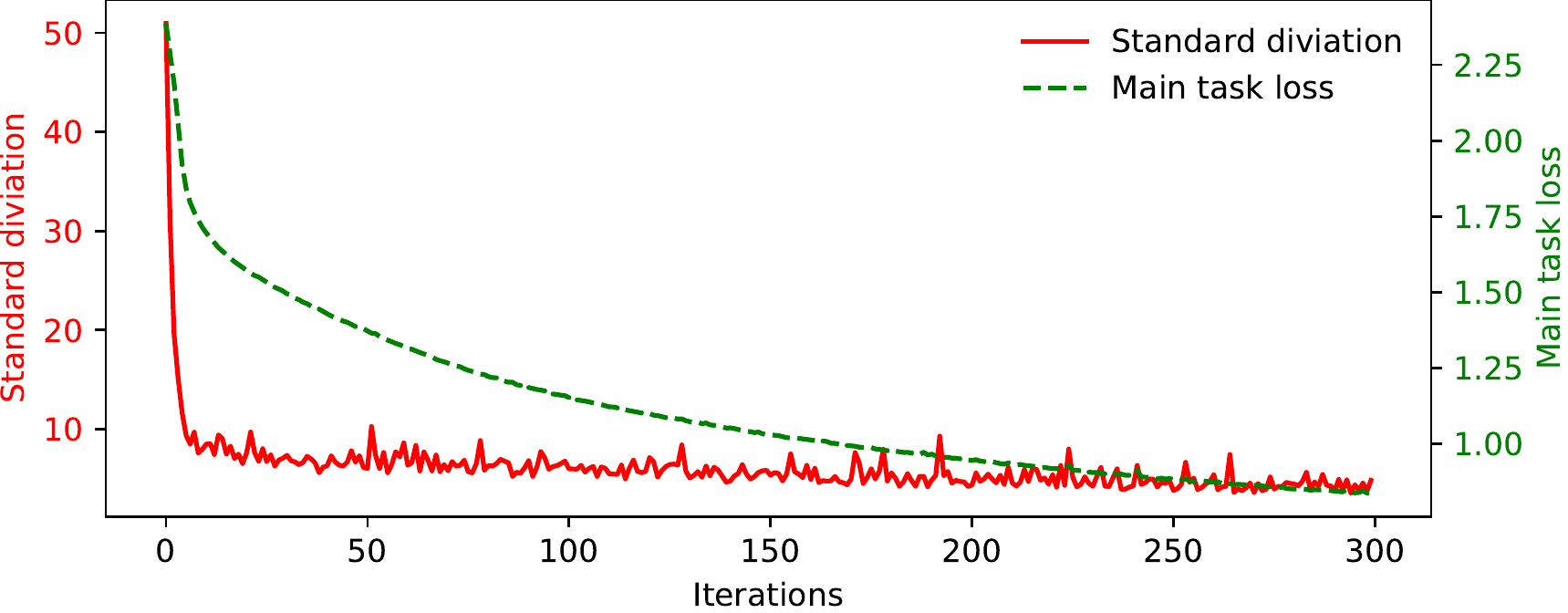}}
\vspace{-3mm}
\caption{SAFELearing when no attacker exists  - main task loss \& standard deviation vs. iterations}\label{fig:STD}
\vspace{-3.5mm}
\end{figure}

One situation is when the attacker continuously manipulates the standard deviation for multiple rounds to neutralize the adaptive threshold by injecting continuous attacks to gradually increase the standard deviation of $\{d_i\}$. However, as discussed in Section \ref{ss:attackstra}, there exist minimal scale factors even for continuous attacks to be effective. Such scaled factors are proportional to the inverse of the learning rate, a hyperparameter controlled by the server. Moreover, without malicious attackers, the standard deviation of $\{d_i\}$ tends to decrease over iterations (even for non-i.i.d data), following the similar pattern of the global model's loss function value. This is because when the global model is converging, the change rate of the global model and contributions from benign users' local models will decrease. This can be validated by our preliminary experiments as shown in Fig. \ref{fig:STD}. Therefore, when continuous attackers first join the training, we can expect a large disparity between the standard deviation of current iteration and previous iterations, which provides opportunities for attack detection. Notes that if the attack starts from the very beginning of training, the main task will be jeopardized. Detail evaluation of such attackers will be present in Section \ref{s:Eval}.


\begin{algorithm}
    \SetAlgoLined
    \KwData{$\overline{X}_L$ is the set of partially revealed aggregated parameter vectors of the subgroup set $G_L$ at layer $L$; $X_t$ is global model of round $t$;
    $Eucl(,)$ computes the Euclidean distance; $Std(\cdot)$ computes the standard deviation; $\xi_t$ is adaptive threshold of round $t$; $\rho$ is the threshold expansion factor; $M=\varnothing$; $D=\varnothing$}
    \KwResult{Attacker-inclusion set $M$}
    
     \For{each vector $\overline{x_i} \in \overline{X}_L$ }{
     \textbf{Compute the distance to global model}
     $d_i=Eucl(H(\overline{x_i}) ,H(X_t))$\\
     $D= D\cup \{d_i\}$\\
     }
    $\xi_t=\frac{\rho}{T}\sum_{j=1}^{T}Std(D)_{t-j}$\\
    \While{$Std(D)>\xi_t$}
    {
        \textbf{Replace the farthest subgroup with global model}
        $H(\overline{x_i}) = H(X_t)$\\
        $D = (D\backslash\{d_i|d_i=Max(D)\})\cup\{0\}$\\
        $M = M\cup \{i\}$\\
    }
    
\caption{Suspicious subgroup detection}\label{Alg:1}
\end{algorithm}

\subsection{Numerical Complexity}

We present overall cost of our protocol in Table \ref{tab:cost}. Detailed complexity analysis and simulation results are elaborated in Appendix \ref{apdx:performance_theo} and \ref{apdx:performace}, respectively.

\begin{table}[!htb]
    \caption{Overall computational and communication costs}\label{tab:cost}
    \vspace{-2mm}
    \centering
    \begin{tabular}{ |m{2.4cm}||m{2.4cm}|m{2.4cm}| } 
        \hline
        \multicolumn{3}{|c|}{Tree Based Secure Aggregation Protocol}\\
        \hline
        & User & Server\\
        \hline
        computation & $O(N+n^2+m)$ & $O(mN+nN)$\\
        communication &$O(N+m)$&$O(N^2+mN)$\\
        \hline
    \end{tabular}
    
    \vspace{-3mm}    
\end{table}

\section{SECURITY ANALYSIS}

We show the security of our protocol with the following theorems, where $t\ge {|U_a|}$ denotes the threshold number of attackers, $C\subseteq U \cup \{S\}$ an arbitrary subset of parties, $x_*$ the input of a party or a set of parties $*$, $U$ the set of all users, $S$ the server and $k$ the security parameter. 

\begin{theorem}\label{TH:3}(Local Model Privacy under Type \rom{1} Attackers)
There exists a PPT simulator $\rm{SIM}$ such that for all $t,U,x_U$ and $C\subseteq U \cup \{S\}$, where $|C\backslash  \{S\}| \leq t$, the output of $\rm{SIM}$ is computationally indistinguishable from the output of  $\rm{Real}_C^{U,t,k}(x_U,U)$\:\newline  $$\rm{Real}_C^{U,t,k}(x_U,U) \approx \rm{SIM}_C^{U,t,k}(x_C,z,U)$$
where 
$$z=
\begin{cases}
\sum_{u\in U\backslash C}x_u & \text{if |U| $\geq $ t} \\
\bot & \text{o.w.}
\end{cases}
$$
\end{theorem}
\begin{proof}
Detailed proof is presented in Appendix \ref{apdx:type1}.
\end{proof}

\begin{theorem}\label{TH:4}(Random Tree Structure Secrecy)
There exists a PPT simulator $\rm{SIM}$ such that for all $t,U,x_U$ and $C\subseteq U \cup \{S\}$, where $|C\backslash  \{S\}| \leq t$, the output of $\rm{SIM}$ is indistinguishable from the output of real protocol:   $$\rm{Real}_C^{U,t,k}(R_u,R_S,C_u^{PK}) \approx \rm{SIM}_C^{U,t,k}(R_u,R_S,C_u^{PK})$$
\end{theorem}

\begin{proof}
We prove this by a standard hybrid argument. The detailed proof is in Appendix \ref{apdx:type1}. 
\end{proof}

\begin{theorem}\label{TH:5} (Indistinguishability of Type \rom{2} attackers in the same subgroup)
For any type \rom{2} attacker $\mathcal{A}$ in a certain subgroup, if there is another type \rom{2} attacker $\mathcal{B}$ in the same subgroup but is not peered with $\mathcal{A}$ for pairwise masking, the joint view of $\mathcal{A}$ and $\mathcal{B}$ $View_{\mathcal{A},\mathcal{B}}$ is indistinguishable from $View_{\mathcal{A}}$, the view of $\mathcal{A}$:   
$$View_{\mathcal{A}} \approx View_{\mathcal{A},\mathcal{B}}$$
\end{theorem}

\begin{proof}
 The detailed proof is in Appendix \ref{apdx:type2}. 
 \end{proof}

\begin{figure*}[hbt!]
\centerline{\includegraphics[width=\linewidth]{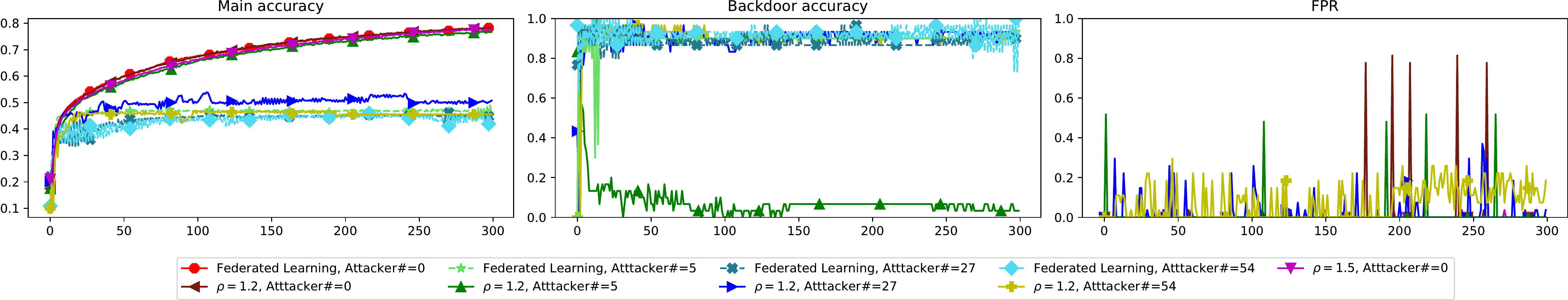}}
\vspace{-4mm}
\caption{Ever-present attack - main task accuracy, backdoor accuracy \& FPR vs. iterations.}\label{fig:Con}
\vspace{-5mm}
\end{figure*}
\section{Backdoor Attack Detection Evaluation}\label{s:Eval}

In this section, we evaluate our protocol against \textbf{Type \rom{2} Attacker} by testing our protocol under state-of-the-art semantic backdoor attack \cite{bagdasaryan2018backdoor} during the training of ResNet-18 network \cite{he2016deep} on CIFAR-10 dataset. The examples of backdoor data are shown in Fig. \ref{fig:sample} and  the effectiveness of targeted attack against conventional secure aggregation \cite{bonawitz2017practical} is shown in Fig. \ref{fig:ori}.

\textbf{Experiment setup:} We simulate a network of $N=1000$ users with non-i.i.d. training data by dividing 50,000 training images using a Dirichlet distribution with hyperparameter 0.9 as in \cite{bagdasaryan2018backdoor}. We assume a $3\times 3$ (heights $\times$ degree) tree which yields 27 subgroups. 

\textbf{Metrics:} We use the following evaluation criteria: (1) \textsl{Detection Rate (DR)}: number of rounds with attackers detected/ number of rounds with attackers presented, (2) \textsl{Correction Rate (CR)}: number of attackers in subgroups that have been labeled as malicious by the server/ total number of attackers, (3) \textsl{False positive rate (FPR)}: number of benign subgroups that have been labeled as malicious by the server/ total number of subgroups and (4) \textsl{After Attack Accuracy:} main \slash backdoor accuracy after the attack
\begin{figure}[hbt!]
\begin{subfigure}[t]{.33\linewidth}
  \centering
  \includegraphics[width=0.9\linewidth]{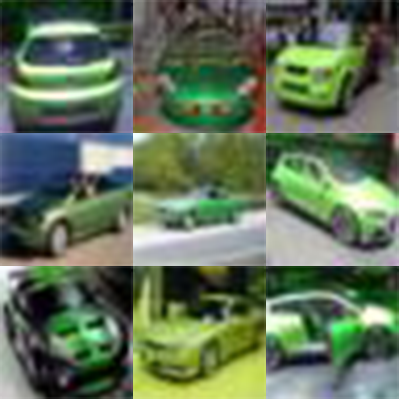}
  \caption{}
\end{subfigure}%
\begin{subfigure}[t]{.33\linewidth}
  \centering
  \includegraphics[width=0.9\linewidth]{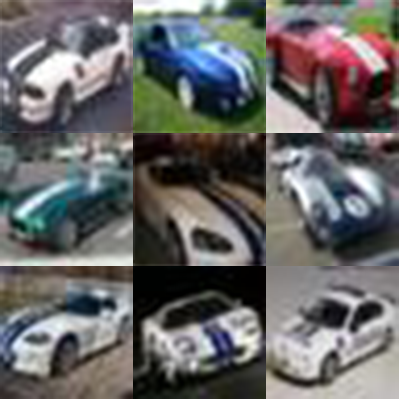}
  \caption{}
\end{subfigure}%
\begin{subfigure}[t]{.33\linewidth}
  \centering
  \includegraphics[width=0.9\linewidth]{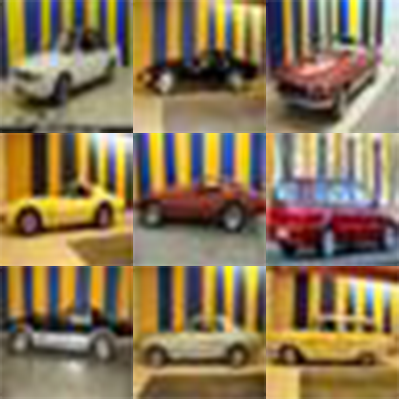}
  \caption{}
\end{subfigure}%
\vspace{-3mm}
\caption{Examples of Backdoor data. Cars with certain attributes are classified as birds. (a) Cars painted in green; (b) Cars with racing stripe and (c) Cars in front of vertical stripes background.}
\label{fig:sample}
\vspace{-4mm}
\end{figure}

\begin{figure}[hbt!]
\centerline{\includegraphics[width=\linewidth]{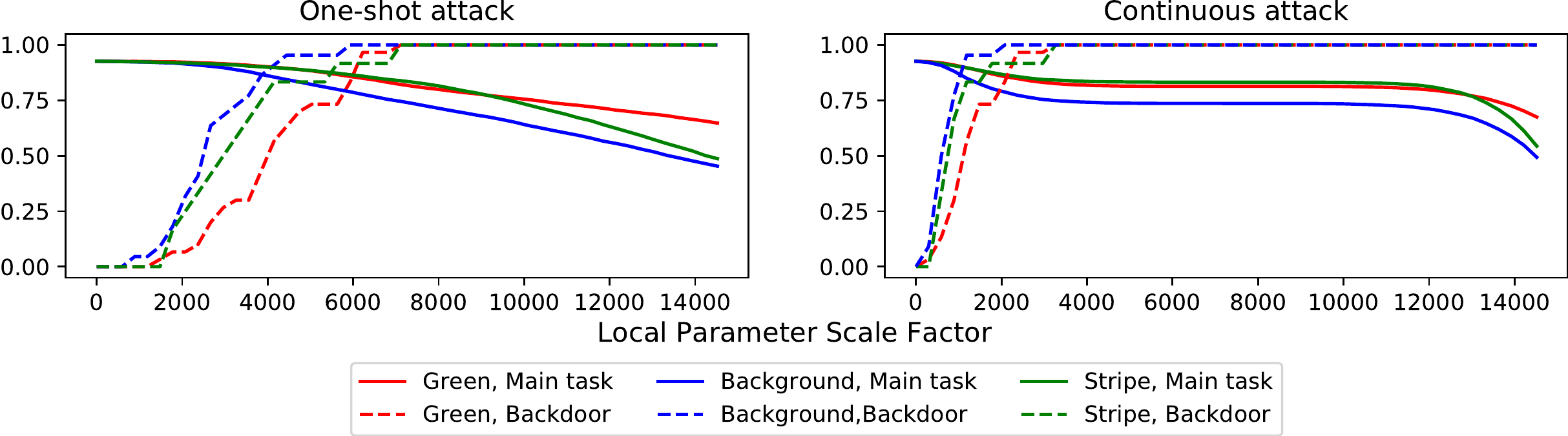}}
\vspace{-3mm}
\caption{After-attack accuracy of three different backdoor data sets of \cite{bagdasaryan2018backdoor} versus the scale factor. 1000 users in total and $2.7\%$  malicious.}
\vspace{-5mm}
\label{fig:ori}
\end{figure}

\begin{figure*}[hbt!]
	\centering
	\begin{subfigure}{0.9\textwidth}
		\includegraphics[width=\linewidth]{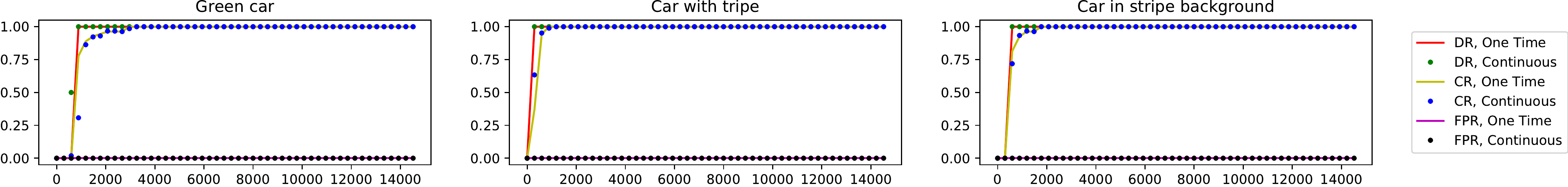}
		\caption{DR, CR and FPR vs. local parameter scale factor.}
	\end{subfigure}
		\begin{subfigure}{0.9\textwidth}
		\includegraphics[width=\linewidth]{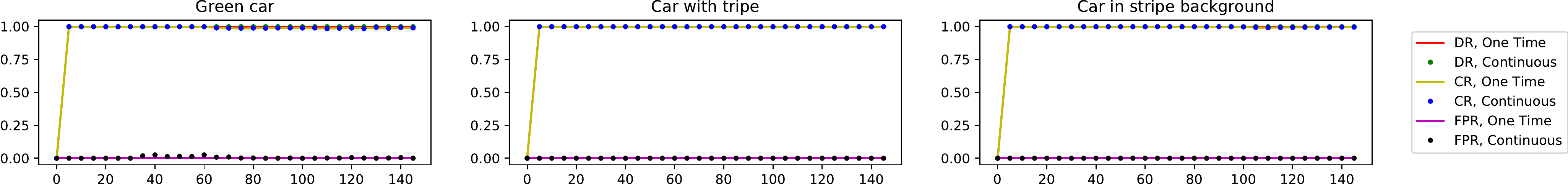}
		\caption{DR, CR and FPR vs. Number of attackers.}
	\end{subfigure}
	\begin{subfigure}{0.9\textwidth}
		\includegraphics[width=\linewidth]{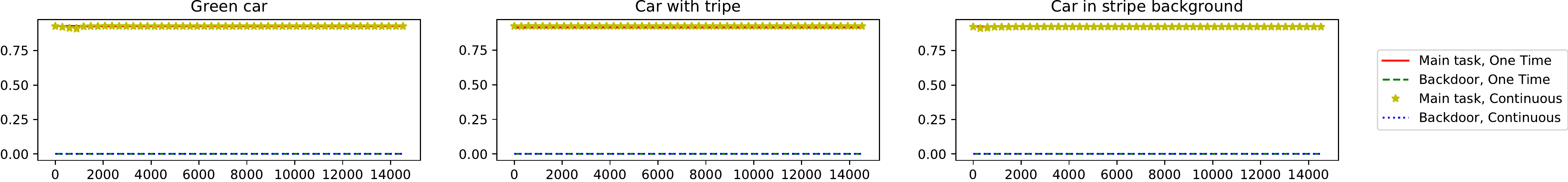}
		\caption{After attack accuracy vs. local parameter scale factor.}
	\end{subfigure}
		\begin{subfigure}{0.9\textwidth}
		\includegraphics[width=\linewidth]{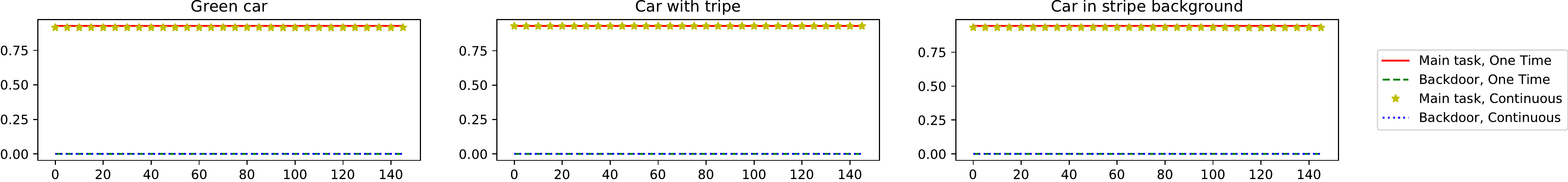}
		\caption{After attack accuracy vs. Number of attackers.}
	\end{subfigure}
	\vspace{-2mm}
	\caption{Performance of our protocol ($\rho=1.2$) against backdoor attack with different backdoor trigger.
	}
    \vspace{-4mm}
	\label{fig:our}
\end{figure*}

\textbf{Performance under ever-present attack:}  In Fig. \ref{fig:Con} we consider the attack that lasts from the very beginning of training to the convergence time, a.k.a. "ever-present attack". We show the main \slash backdoor accuracy and FPR as a function of the number of iterations of training by considering different attacker numbers and thresholds. When no attacker exists, our protocol achieves almost the same training efficiency as conventional federated learning. By comparing our protocol with different threshold expansion factors $\rho$, we note that a tighter threshold will cause false positives occasionally, but the overall training efficiency is not influenced. When ever-present attackers are presented, although the attackers can easily inject backdoor to the global model, the training process is stagnated because the attacker will eliminate updates from benign users as shows in (\ref{eq:atk2}). On the other hand, when the number of attackers is smaller than the number of subgroups (e.g., attacker$\#=5$), our protocol can eliminate all attacker excepts for first few rounds when the standard deviation of updates from benign users is high, and the early injected backdoor will be "forgotten" while the main task converges normally. Although our protocol cannot avoid backdoor injection when a large number of ever-present attackers are presented, the attackers cannot achieve the attack goal either due to the failure of main task convergence.

\textbf{Performance against converging phase attacker:} \cite{xie2019dba}\cite{bagdasaryan2018backdoor} concludes that launching the attack toward the convergence time is beneficial to attackers because the backdoor tends to be "forgotten" easily in early stages and the main task may be jeopardized, which consists with our results above. In Fig. \ref{fig:our}, we show the performance of our protocol against the one-shot attacker (attacking at the last iteration of convergence) and continuous attacker (continuously attacking for 5 iterations before convergence) launched in the converging phase. Our protocol achieves $100\%$ detection rate, near $100\%$ correction rate, and $0\%$ false positive rate in most of the scenarios, which supports our analysis in Section \ref{ss:BAD} that there will be significant distribution disparity between benign subgroup and subgroups with attackers. After-attack accuracy shows that neither the main task was influenced nor the backdoor was injected. However, when the number of attackers exceeds the number of subgroups, the training process is stagnated because most of the subgroups will be replaced.

\section{Related Work}\label{s:Related}

\subsection{Model-Poisoning Attacks to Distributed Machine Learning} In distributed machine learning users and the aggregation server, if any, do not have access to the training data sets possessed by other users. This provides the opportunity for various malicious attacks. Data poisoning  \cite{biggio2012poisoning} and model replacement  \cite{bagdasaryan2018backdoor} are two most common attacks that aim at generating malicious global models at the attacker's will. Specifically, traditional data poisoning attacks \cite{chen2017targeted,gu2017badnets,liu2017trojaning}, which usually target at cloud-centric learning, can influence the behavior of the global model by constructing poisoning samples and uploading poisoned features. According to recent study, the attacker needs to pollute about $20\%$ to $40\%$ of the training data in targeted classes \cite{biggio2012poisoning} or focus on training data with rare features \cite{huang2011adversarial} in order to launch data poisoning attacks. In large-scale machine learning the attacker usually needs to compromise $10\%$ to $50\%$ of the participants \cite{Shen2016A,blanchard2017machine} who continuously upload malicious models. 


Model replacement attackers, on the other hand, leverage the information of the global model and locally construct malicious inputs that modify the global model precisely in the way they desire. To make the attack more effective, attackers can adjust and augment their local model parameters to dominate the global model during the aggregation process. As compared to data poisoning attacks, model poisoning is more efficient and effective. The objective of the attack, even if it is to completely replace the global model, can be achieved in one shot with one or few attackers. As sufficient validation datasets are not always available to the aggregator, detection of such attacks is nontrivial. Existing model replacement detection techniques like Byzantine-robust federated learning \cite {blanchard2017machine} address the problem by evaluating the consistency of models provided by all participating users, either by amplitudes \cite {blanchard2017machine} or by angles \cite{fung2018mitigating} of the received gradient updates. As these techniques need to evaluate models provided by individual users at each iteration, a prohibiting complexity is introduced in large-scale systems with many participants. Moreover, all the above detection techniques require access to plaintext of individual models which is not available in secure aggregation. 

\subsection{Privacy Preserving Machine Learning} Privacy-preserving machine learning \cite{RC-2019} aims to prevent or constrain disclosure of training data or models to unauthorized parties. To this end various techniques have been proposed which can be roughly categorized into secure multi-party computation  (MPC) and differential privacy, based on the underlying techniques they employ. In the area of distributed machine learning, existing MPC-based proposals usually rely on heavy cryptographic primitives include gable circuits, partially or fully homomorphic encryption, oblivious transfer, etc. Recently, promising progresses have been made toward small networks especially for inference tasks \cite{mohassel2017secureml, SecureNN19, GAZELLE18, ABY3-18,Chameleon18,MiniONN17, nikolaenko2013privacy, bost2015machine}. However, there has yet been a practical cryptographic tool that supports efficient training of complex models, e.g., for deep learning tasks.

Differential privacy (DP), on the other hand, focuses on publishing aggregated information with limited disclosure of private information. For example, one approach \cite{abadi2016deep,geyer2017differentially} is to protect data\slash model\slash outputs by allowing users to add zero-mean statistical noise (e.g., Laplace noise) to make data\slash model\slash outputs indistinguishable and aggregate data\slash model\slash outputs with affordable variance. To maintain a desired accuracy, one needs to carefully design the random noise without degrading the level of privacy protection when the DP mechanism is repeated in the training process. As a result, it remains a challenge in DP to maintain an appropriate trade-off between privacy and model quality ( in terms of accuracy loss caused by added noise) especially in deep learning tasks.

\section{Conclusion}\label{s:Conclusion}
This paper presented a novel secure aggregation scheme for federated learning, which supports backdoor detection and secure aggregation simultaneously with oblivious random grouping and partial parameter disclosure. Compare to conventional secure aggregation protocol, our protocol reduces computation complexity from $O(N^2+mN)$ (at user side) and $O(mN^2)$ (at server-side) to $O(N+n^2+m)$ and $O(mN+nN)$, respectively, where $N$ is the total number of users, $n$ the number of users in subgroups and $m$ the size of the model. We validated our design through experiments with 1000 simulated users. Experimental results demonstrate the efficiency and scalability of the proposed design. Moreover, our protocol utilizes adaptive outlier detection to address potential backdoor attack in federated learning. Experimental evaluation on state-of-art adaptive backdoor attack attack\cite{bagdasaryan2018backdoor} combining with possible attack strategies such as continuous attacks show that our protocol can significantly improve the robustness against backdoor attacks.

\section{Acknowledgement}
This project is partially supported by NSF grant ECCS\#1923739.

\bibliographystyle{acm}
\bibliography{citation}

\appendix

\appendix

\section{Proof of Theorem \ref{TH:2}} \label{app:theo2}

\textbf{Theorem \ref{TH:2}}: Let $\{X'_i\}_{i\in [1,n]}$ be $n$ samples randomly selected from $\{X_i\}_{i\in [1,N]}$ in {Theorem \ref{TH:1}} and $Y=y\in [\overline{\mathbb{E}}(X')-R_H/2,\overline{ \mathbb{E}}(X')+R_H/2]$, where $\overline{ \mathbb{E}}(X')$ is the mean vector of $\{X'_i\}_{i\in [1,n]}$. If $R_H=2(1-\sqrt{\frac{n-1}{n}})\varepsilon$. The expectation of the bound of $P(|X_i-y|<\sigma|Y)$ is $\frac{\overline{\sigma}_N^2}{\varepsilon^2}$.

\begin{proof}
Assume $\{X'_i\}_{i\in 1\to n}$ have variance $\overline{\sigma}_n^2$. Using Chebyshev Inequality, we can gets:
\begin{align*}
    P(||X_i-\overline{ \mathbb{E}}(X')||\geq \varepsilon)\leq \frac{\overline{\sigma}_n^2}{\varepsilon},\varepsilon^2>\frac{R_H}{2}\\
\end{align*}
Combining with $Y=y\in [\overline{\mathbb{E}}(X')-R_H/2,\overline{ \mathbb{E}}(X')+R_H/2]$, we can derive the bounds:
\begin{align*}
    P(||X_i-y||\geq \varepsilon&|y\in [\overline{\mathbb{E}}(X')-R_H/2,\overline{ \mathbb{E}}(X')+R_H/2])\\
    &\leq \frac{\overline{\sigma}_n^2}{(\varepsilon-\frac{R_H}{2})^2},\varepsilon>\frac{R_H}{2}\\
\end{align*}
Because $\{X'_i\}_{i\in 1\to n}$ are randomly sampled from $\{X_i\}_{i\in 1\to N}$, according to the distribution of sample variance \cite{walpole1993probability}, the expectation of $\overline{\sigma}_n^2$ is $\frac{n-1}{n}\overline{\sigma}_N^2$. Therefore the expectation of the bounds can be written as:
\begin{align*}
        E(\frac{\overline{\sigma}_n^2}{(\varepsilon-\frac{R_H}{2})^2})=\frac{(n-1)\overline{\sigma}_N^2}{n(\varepsilon-\frac{R_H}{2})^2}
\end{align*}
When we let $R_H=2(1-\sqrt{\frac{n-1}{n}})\varepsilon$ the expectation will be $\frac{\overline{\sigma}_N^2}{\varepsilon^2}$, which is the same in {Theorem \ref{TH:1}}.
\end{proof}

\section{Type \rom{1} Attacker}\label{apdx:type1}

Here we give the formal proof of Theorem \ref{TH:3}. We have the following parameters setting in our protocol execution:  a security parameter $k$ which the underlying cryptographic primitives are instantiated with, a threshold $t$ to limit the number of corrupted parties, a server $S$ and a user set $U$ which contains $N$ users in total, each one is denoted as $u_i$ where $i=\{1,2,...,N\}$. During the execution, we require that the total dropouts to be limited by certain threshold. We denote the input of each user $u$ as $x_u$ and $x_{U'} = \{x_u\}_{u\in U'}$ as the inputs of any subset of users $U'\subseteq U$.

In the protocol execution process, the view of a party consists of its inputs and the randomness and all messages this party received from other parties. If a party drop out during the process, its view remains the same as that before the dropout occurs. 

Given any subset $C\subseteq U \cup \{S\}$ of the parties, let $\rm{Real}_C^{U,t,k}(x_U,U)$ be a random variable which represents the joint view of all parties in $C$ which includes some honest but curious users and an honest but curious server who can combine knowledge of these users. Our theorem shows that the joint view of any subset of honest but curious parties can be simulated given the inputs of these part of users and only the sum of the input from those users within the same group. In particular, we prove that with a certain threshold $t$, the joint view of the server and any set of less than $t$ corrupted party will not leak any information of other honest parties besides the output which is known to all the parties and server. 

\textbf{Theorem \ref{TH:3}}  (Local Model Privacy under Type \rom{1} Attackers)
There exists a PPT simulator \rm{SIM} such that for all $t,U,x_U$ and $C\subseteq U \cup \{S\}$, where $|C\backslash  \{S\}| \leq t$, the output of \rm{SIM} is computationally indistinguishable from the output of $\rm{Real}_C^{U,t,k}(x_U,U)$\:\newline  $$\rm{Real}_C^{U,t,k}(x_U,U) \approx \rm{SIM}_C^{U,t,k}(x_C,z,U)$$
where 
$$z=
\begin{cases}
\sum_{u\in U\backslash C}x_u & \text{if |U| $\geq $ t} \\
\bot & \text{o.w.}
\end{cases}
$$

\begin{proof}We prove the theorem by a standard hybrid argument. We will define a simulator \rm{SIM} through a series of subsequent modifications to the random variable \rm{Real}, so that any two consecutive random variables are computationally indistinguishable.

$Hybrid_0$ This random variable is distributed exactly as \rm{Real}, the joint view of the parties $C$ in a real execution of the protocol.

$Hybrid_1$ In this hybrid, we change the behavior of those honest parties in simulator. Instead of using the secret sharing keys $(C_u^{SK},C_v^{PK})$  produced by Diffie-Hellman key exchange, these parties use a uniformly random encryption key $c_{u,v}$ chosen by the simulator. This hybrid is computationally indistinguishable with the previous one guaranteed by our Server Involved Decisional Diffie-Hellman key exchange.

Recall that in our Server Involved D-H key exchange, the only modification we made compared to the D-H key exchange in original protocol is that, we add a randomness to the public key. In specific, instead of exchanging the public key $\{s_u^{PK},s_v^{PK}\}$, user $u$ and $v$ will exchange $\{(s_u^{PK})^{r_{u,v}},(s_v^{PK})^{r_{u,v}}\}$. We show that for any honest but curious parties who is given this pair of public key can not distinguish the secret $s_{u,v}$ computed by these keys from a uniformly random string. 

Assume there exists a PPT algorithm $\mathcal{A}$ which can break the server involved D-H key exchange with advantage $\epsilon$, then there exists an a PPT algorithm $\mathcal{B}$ which can break DDH assumption with the same advantage. The proof follows the setting in original protocol. Consider a multiplicative cyclic group $G$ of order $q$, and with generator $g$, for $a,b,c$ uniformly random chosen in $\mathbb {Z}_q$ and a randomness $r_{a,b}$ produced by $VRF$, $\mathcal{A}$ can distinguish $g^{{ab}^{r_{a,b}}}$ from $g^c$, then for $\mathcal{B}$, it times the result of D-H key exchange for $r_{a,b}$ times, that is $(g^{ab})^{r_{a,b}}$ and then invoke $\mathcal{A}$, it breaks the DDH with the same advantage.

$Hybrid_2$ In this hybrid, we change the ciphertexts encrypted by the honest parties. Instead of using the encryption of correct shares $S_u^{SK}$ and $b_u$, each honest party just encrypt $0^*$ which has the same length as those shares and send it to other honest parties. This hybrid is computationally indistinguishable with the previous one. This is guaranteed by the IND-CPA security of secret-sharing scheme.

$Hybrid_3$ In this hybrid, we change the shares of $b_u$ given to the corrupted parties to shares of 0. Note that in our design, the view of the corrupted parties contains no more than $\lceil t/n\rceil$ shares of each $b_u$,since the honest users do not reveal $b_u$ for corrupted parties during Unmasking procedure. This hybrid is distributed the same as the previous one. This is guaranteed by the properties of Shamir's secret sharing scheme.

$Hybrid_4$ In this hybrid, we change a part of pariwise masking information of all the parties in $U$. Instead of computing $PRG(b_u)$, the simulator use a uniformly random vector of the same size. Note that in previous hybrid, since $b_u$ is uniformly random and is substituted by shares of 0,the output does not depend on the seed of $PRG$. Thus, when we substitute the result of $PRG$ with a random value. the result is guaranteed to have the same distribution by the security of $PRG$ which means this hybrid is distributed the same as the previous one.

$Hybrid_5$ In this hybrid,we substitute the masked input without the input of the users. That is,instead of sending
\begin{align}
&y_u = x_u+PRG(b_u)+ 
    \smashoperator[r]{\sum_{\forall v\in G_u^s:u<v}}PRG(s_{u,v})-\nonumber\\ 
    &\smashoperator[r]{\sum_{\forall v\in G_u^s:u>v}}PRG(s_{v,u})+
    \smashoperator[r]{\sum\limits_{\forall v\in G^p_u:u<v}}(\Lambda_{R_H}\land PRG(s_{u,v}))-\nonumber\\
    &\smashoperator[r]{\sum\limits_{\forall v\in G^p_u:u>v}}(\Lambda_{R_H}\land PRG(s_{v,u}))\quad (mod\quad R)\label{eq:mask}
\end{align}
users send: 
\begin{align}
&y_u = PRG(b_u)+ 
    \smashoperator[r]{\sum_{\forall v\in G_u^s:u<v}}PRG(s_{u,v})-\nonumber\\ 
    &\smashoperator[r]{\sum_{\forall v\in G_u^s:u>v}}PRG(s_{v,u})+
    \smashoperator[r]{\sum\limits_{\forall v\in G^p_u:u<v}}(\Lambda_{R_H}\land PRG(s_{u,v}))-\nonumber\\
    &\smashoperator[r]{\sum\limits_{\forall v\in G^p_u:u>v}}(\Lambda_{R_H}\land PRG(s_{v,u}))\quad (mod\quad R)\label{eq:mask}
\end{align}
Recall that in the previous hybrid, $PRG(b_u)$ was changed to be uniformly random, so $x_u+PRG(b_u)$ is also uniformly random. Thus, this hybrid is identical to the previous one and further more, all the following hybrids do not depend the values of $x_u$

$Hybrid_6$ In this hybrid, we change the behavior by sending 0 instead of $S_u^{SK}$ generated by the honest parties to all other parties. Similar with $Hybrid_3$,the properties of Shamir's secret sharing scheme guarantee that this hybrid is identical to the previous one.  

$Hybrid_7$In  this hybrid, for a specific honest user $u'$, in order to compute $y_u$, we substitute the key $s_{u,v}$ and $s_{v,u}$ for all other users in $G^p_u$ with two uniformly random value. This hybrid is computationally indistinguishable with the previous one. This is guaranteed by the Decisional Diffie-Hellman assumption.

$Hybrid_8$ In this hybrid, instead of using $PRG(s_{u,v})$ and $PRG(s_{v,u})$ in $\Lambda_{R_H}\land PRG(s_{u,v})$ and $\Lambda_{R_H} \land  PRG(s_{v,u})$, we compute $y_u$ with two uniformly random values with the same size. Recall that The bits of $\Lambda_{R_H}$'s elements are $0$ if they are in range $[R_H,R_U]$, other wise the bits are $1$. This operation will not affect the randomness since all the values are doing an and operation with 0 on high bits, this hybrid is indistinguishable with the previous one. This is guaranteed by the security of $PRG$.

$Hybrid_9$ In  this hybrid, for a specific honest user $u'$, in order to compute $y_u$, we substitute the key $s_{u,v}$ and $s_{v,u}$ for all other users in $G^s_u$ with two uniformly random value. This hybrid is computationally indistinguishable with the previous one. This is guaranteed by the Decisional Diffie-Hellman assumption.

$Hybrid_{10}$ In this hybrid, instead of using $PRG(s_{u,v})$ and $PRG(s_{v,u})$, we compute $y_u$ with two uniformly random values with the same size. Similar to $Hybrid_4$, this hybrid is distinguishable with the previous one. This is guaranteed by the security of $PRG$.

$Hybrid_{11}$ In this hybrid, for all the honest users,instead of sending
\begin{align}
&y_u = x_u+PRG(b_u)+ 
    \smashoperator[r]{\sum_{\forall v\in G_u^s:u<v}}PRG(s_{u,v})-\nonumber\\ +
    &\smashoperator[r]{\sum_{\forall v\in G_u^s:u>v}}PRG(s_{v,u})+
    \smashoperator[r]{\sum\limits_{\forall v\in G^p_u:u<v}}(\Lambda_{R_H}\land PRG(s_{u,v}))-\nonumber\\
    &\smashoperator[r]{\sum\limits_{\forall v\in G^p_u:u>v}}(\Lambda_{R_H}\land PRG(s_{v,u}))\quad (mod\quad R)\label{eq:mask}
\end{align}
users send
\begin{align}
&y_u = w_u+PRG(b_u)+ 
    \smashoperator[r]{\sum_{\forall v\in G_u^s:u<v}}PRG(s_{u,v})-\nonumber\\ 
    &\smashoperator[r]{\sum_{\forall v\in G_u^s:u>v}}PRG(s_{v,u})+
    \smashoperator[r]{\sum\limits_{\forall v\in G^p_u:u<v}}(\Lambda_{R_H}\land PRG(s_{u,v}))-\nonumber\\
    &\smashoperator[r]{\sum\limits_{\forall v\in G^p_u:u>v}}(\Lambda_{R_H}\land PRG(s_{v,u}))\quad (mod\quad R)\label{eq:mask}
\end{align}
where $w_u$ is uniformly random and satisfy $\sum\limits_{u\in U\backslash C}{w_u}=\sum\limits_{u\in U\backslash C}{x_u}=z$. This hybrid is identically distributed as previous one.

Thus, we can define a PPT simulator $\rm{SIM}$ as the last hybrid describe and the output of this simulator is computationally indistinguishable from the output of $\rm{Real}$.
\end{proof}

In the next part we prove that for our random tree structure under honest but curious setting, the final order output by the protocol is indistinguishable from a uniformly random value in the co-domain of selected hash function.

\textbf{Theorem \ref{TH:4}} (Random Tree Structure Secrecy)
There exists a PPT simulator \rm{SIM} such that for all $t,U,x_U$ and $C\subseteq U \cup \{S\}$, where $|C\backslash  \{S\}| \leq t$, the output of $\rm{SIM}$ is indistinguishable from the output of real protocol:   $$\rm{Real}_C^{U,t,k}(R_u,R_S,C_u^{PK}) \approx \rm{SIM}_C^{U,t,k}(R_u,R_S,C_u^{PK})$$

\begin{proof}
we show this by a standard hybrid argument. We will define a simulator \rm{SIM} through a series of modifications to the random variable \rm{Real}, so that any two consecutive random variables are indistinguishable.

$Hybrid_0$ This hybrid is exactly the same as $Real$.

$Hybrid_1$ In this hybrid, we fix an specific honest user $u'$ and change its behavior. Instead of sending $c_u^{PK}$, $u'$ sends a uniformly random key with the appropriate size to the server. Notice that $Id_u = HASH(R_s||c^{PK}_{u}||R_u)$, thus this hybrid is distributed exactly the same as the previous one. This is guaranteed by the properties of hash function.

$Hybrid_2$ In this hybrid, for the specific user defined in previous hybrid, the final output $HASH(\smashoperator[r]{\sum_{\forall v\in G:v\neq u}} Id_v)$ is indistinguishable from the previous one. It is easy to show that if there exists one more honest user except $u'$, the result will guarantee to uniformly distribute by the properties of hash function.
\end{proof}

Moreover, we can prove that the result still holds for malicious users and compromised server under the security of the commitment protocol. After commitment, these corrupted users and server could not make any change on $R_s$ and $R_u$, and then the result of $Id_u$ is indistinguishable from a random value in co-domain of hash function by the properties of hash function. 

\section{Type \rom{2} Attacker}\label{apdx:type2}
 
In this part we give a formal statement to our claim. For a fixed attacker $\mathcal{A}$ in a certain subgroup, we show that even if there is another attacker $\mathcal{B}$ existing in the same subgroup, they can not know they are in the same subgroup which means that the joint view of $\mathcal{A}$ and $\mathcal{B}$ is indistinguishable from the view of $\mathcal{A}$.

\textbf{Theorem 5} (Indistinguishability of Type \rom{2} attackers in the same subgroup)
For any type \rom{2} attacker $\mathcal{A}$ in a certain subgroup, if there is another type \rom{2} attacker $\mathcal{B}$ in the same subgroup but is not peered to $\mathcal{A}$ , $View_{\mathcal{A},\mathcal{B}}$, the joint view of $\mathcal{A}$ and $\mathcal{B}$ is indistinguishable from $View_{\mathcal{A}}$, the view of $\mathcal{A}$:   
$$View_{\mathcal{A}} \approx View_{\mathcal{A},\mathcal{B}}$$

\begin{proof}
We prove the theorem by a standard hybrid argument. Recall that there is $n$ users in a subgroup and each user can only pairwise mask with previous and next $\kappa$ users according to $Id_u$ in our circle-like topology design. We first fix the position of a given isolated attacker $\mathcal{A}$  which is peered with $\kappa$ honest user in this subgroup and an unassigned attacker $\mathcal{B}$.

$Hybrid_0$ This hybrid is exactly the same as our setting. In this hybrid, the view of attacker $\mathcal{A}$ contains the transcript produced by running the protocol with all the pairwise mask peers of $\mathcal{A}$ and all the key pairs owned by the attackers besides the key pairs shared by the secret sharing scheme since we allows the cooperation of attackers.

$Hybrid_1$ In this hybrid, we substitute all the information of these $\kappa$ honest user with random strings and the view of $\mathcal{A}$ remains the same. This part of proof follows the same as our proof for type \rom{1} attacker and we omit it.

$Hybrid_2$ In this hybrid, for a specific user $u$ who is not a peer of $\mathcal{A}$ and let $v$ denote the pairwise mask peers of $u$, we substitute the key $c_u^{SK}$ with $c_{\mathcal{B}}^{SK}$. This modification will not change the view of $\mathcal{A}$ since $\mathcal{A}$ has no information about $c_u^{SK}$ guaranteed by the fact that the server will only send the the public key from user $u$'s pairwise companions.

$Hybrid_3$ In this hybrid, instead of using $s_u^{SK}$ and $b_u$, user $u$ shares the share of $\mathcal{B}$. Similar with $Hybrid_2$, the server will only relay the encrypted secret shares to the peers of user $u$. Thus, for attacker $\mathcal{A}$, the view remains the same. 

$Hybrid_4$ In this hybrid, we substitute the pairwise key pairs $(s_v^{SK},s_u^{PK})$ by the keys of $\mathcal{B}$. Similarly, this will not change the view of $\mathcal{A}$. The view of $\mathcal{A}$ has already contains the key of $\mathcal{B}$ and the circle-like topology guarantee that $\mathcal{A}$ cannot execute the pairwise masking with $u$ so they share no information about the pairwise mask keys. Till now, we have substitute all the information of an honest user $u$ by the information of an attacker $\mathcal{B}$. The joint view of $\mathcal{A}$ and $\mathcal{B}$ contains the transcripts produced by running the protocols. Since we have already prove that the protocol is privacy-preserving. The joint  view of $\mathcal{A}$ and $\mathcal{B}$ makes no difference with joint view of $\mathcal{A}$ and an honest user $u$. Thus, if two attackers are not peered with each other for pairwise masking, any of them cannot get any extra information more than the information it owns.
\end{proof}

\section{Performance analysis: THEORETICAL}\label{apdx:performance_theo}

 In the following analysis, we use $h$ and $d$ to indicate the height and degree of the tree structure, respectively.

\textbf{User's Computational Cost:} $O(N+n^2+m)$. Each user $u$'s computational cost includes: (1) performing $n+2(\kappa+\log\frac{N}{n})$ key agreements of complexity $O(n)$; (2) creation of $t$-out-of-$n$ Shamir's secret shares of $s^{SK}_u$ and $b_u$, which is $O(n^2)$; (3) verification of random full order generation, which is $O(N)$; and (4) performing intra-group masking and inter-group masking on data vector $x_u$, which is $O(m(\kappa+h))=O(m)$. The overall computational cost for each user is $O(N+n^2+m)$.

\textbf{User's Communication Cost:} $O(N+m)$. Each user $u$'s communication cost includes: (1) exchanging keys with all other users in the verification phase, which accounts for $O(N)$; (2) transmitting constant parameters like $R_u$ which is $O(1)$; (3) sending and receiving $2(n-1)$ encrypted secret shares of constant size, which is $O(n)$; (4) sending a masked data vector of size $m$ to the server, which is $O(m)$; (5) sending the server $n$ secrete shares, which is $O(n)$. The overall communication cost for each user is $O(N+m)$.


\textbf{Server's Computational Cost:} $O(mN+nN)$. The server's computational cost includes: (1) reconstructing $N$ $t$-out-of-$n$ Shamir's secrets, which takes a total time of $O(Nn)$ (Note that for each user in the same subgroup, the Lagrange basis polynomials for reconstructing Shamir's secrets remain the same.); (2) random full order generation, which is $O(N)$; (3) for each dropped user, the server needs to unmasking its mask for all survived users in the same group. The expected time cost yields to $O(md(\kappa+h))$, where $d$ is the number of dropped users. For each survived user, the server needs to unmasking its own mask. And the time cost is $O(m(N-d))$. The overall computational cost for the server is $O(mN+nN)$.
\begin{figure}[h]
	\centering
	\begin{subfigure}{\linewidth}
		\captionsetup{width=0.9\textwidth}
		\centering
		\includegraphics[width=\linewidth]{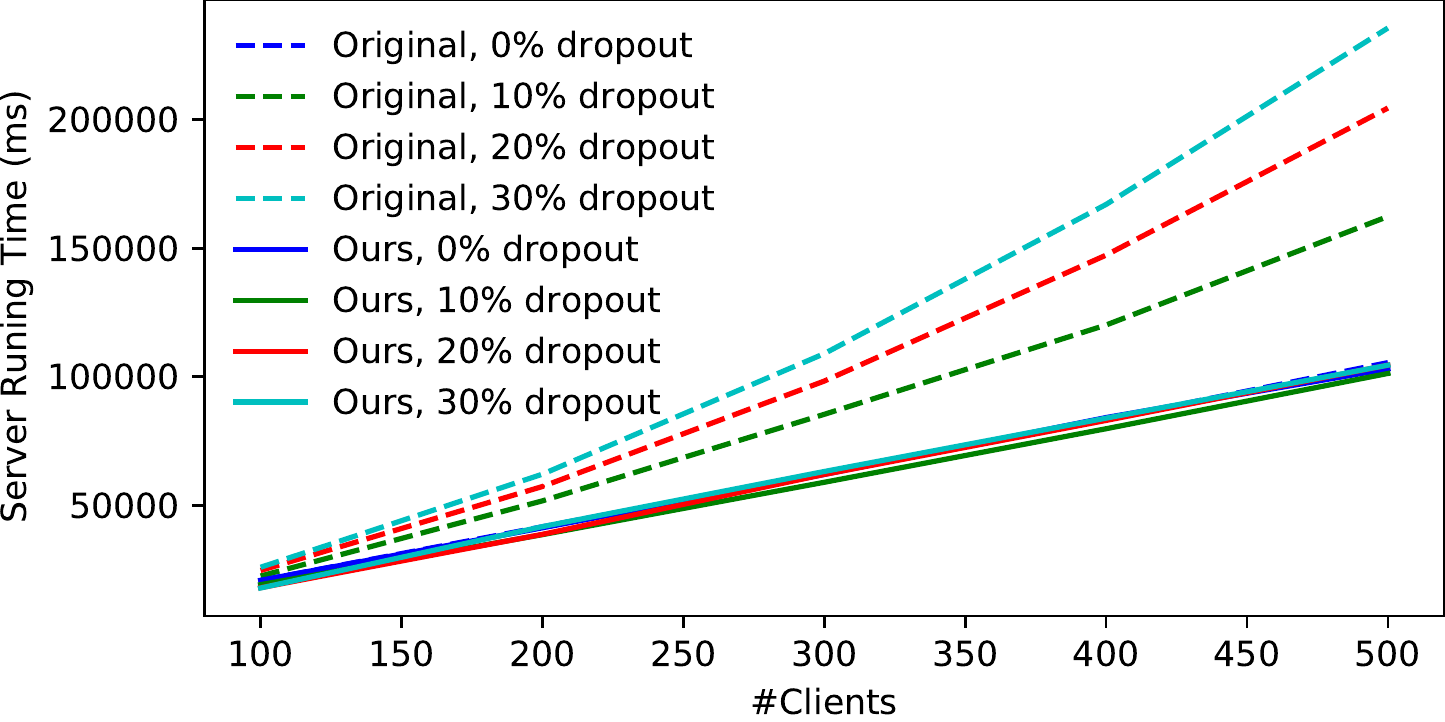}
		\caption{Running time for the server in different clients number. The data vector size is fixed to 100K entries.}
	\end{subfigure}
	\begin{subfigure}{\linewidth}
		\captionsetup{width=0.9\textwidth}
		\centering
		\includegraphics[width=\linewidth]{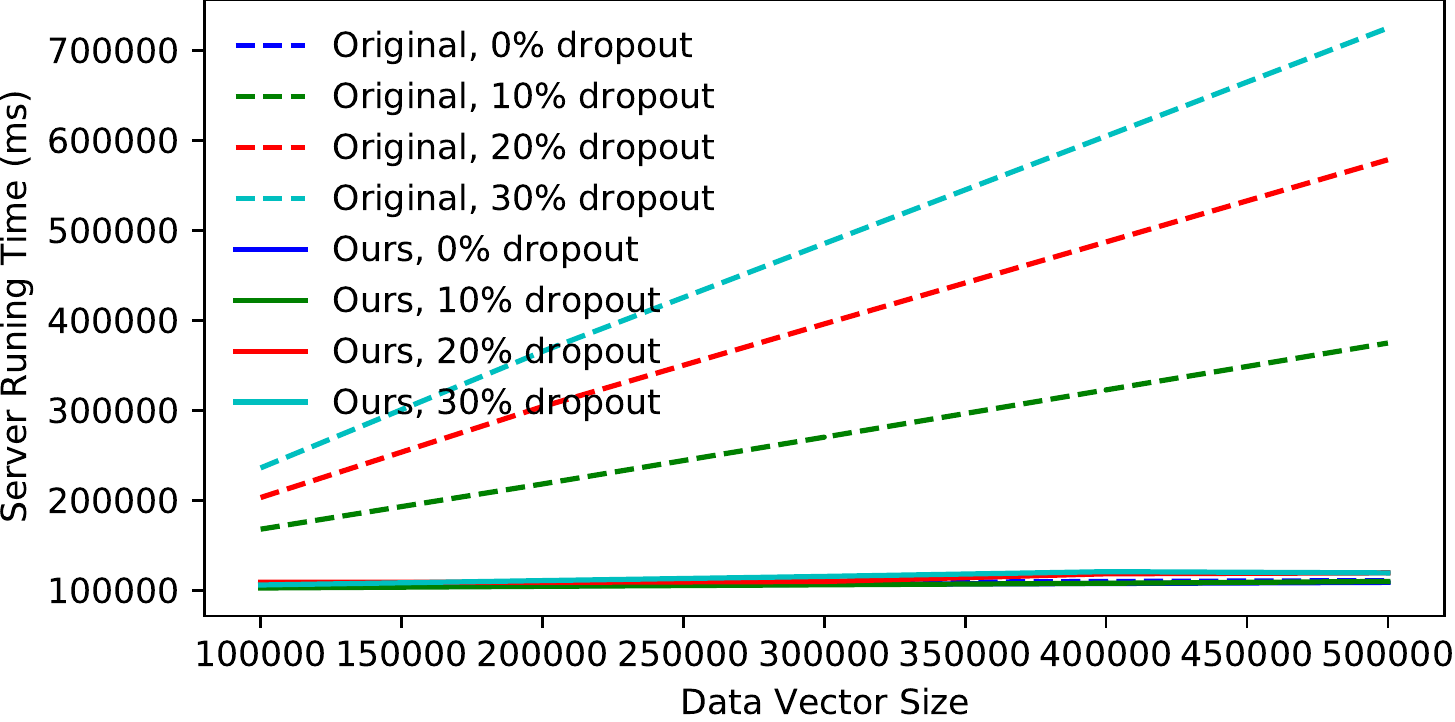}
		\caption{Running time for the server, as the size of the data vector increases. The number of clients is fixed to 500.}
	\end{subfigure}
	\caption{Server Running Time. Solid line represent our Tree Based Secure Aggregation Protocol. Dash line represent original Secure Aggregation Protocol}
	\label{f:server}
\end{figure}
\textbf{Server's communication Cost:} $O(N^2+mN)$. The server's communication cost is mainly dominated by (1) its mediation of pairwise communications between the users, which is $O(N^2)$ and (2) receiving masked data vectors from each user, which is $O(mN)$. The overall communication is $O(N^2+mN)$.
\begin{table*}
	\centering
	\captionsetup{width=0.8\textwidth}
	\caption{users' running time, total data transfer and server running time in different tree structure. The data vector size is fixed to 100K entries, each entry is 10 bytes and the drop rate is fixed to $15\%$.}\label{tab:tree}
	\vspace{-2mm}
	\begin{tabular}{ |c|c|c|c|c| } 
		\hline
		$\#$users& Tree structure(heights$\times$degree) & Running time per user& Total data transfer per user& Server running time\\
		\hline
		1000&$2\times2$&1166 ms&1.51 MB&209283 ms\\
		\hline
		1000&$3\times2$&568 ms&1.40 MB&207076 ms\\
		\hline
		1000&$3\times3$&197 ms&1.33 MB&190610 ms\\
		\hline
		1500&$2\times2$&1877 ms&1.77 MB&312432 ms\\
		\hline
		1500&$3\times2$&899 ms&1.61 MB&310422 ms\\
		\hline
		1500&$3\times3$&298 ms&1.50 MB&287882 ms\\
		\hline
	\end{tabular}
\end{table*}
\begin{figure*}[h]
	\centering
	\begin{subfigure}[t]{0.30\textwidth}
		\captionsetup{width=0.9\textwidth}
		\centering
		\includegraphics[height=1.5in]{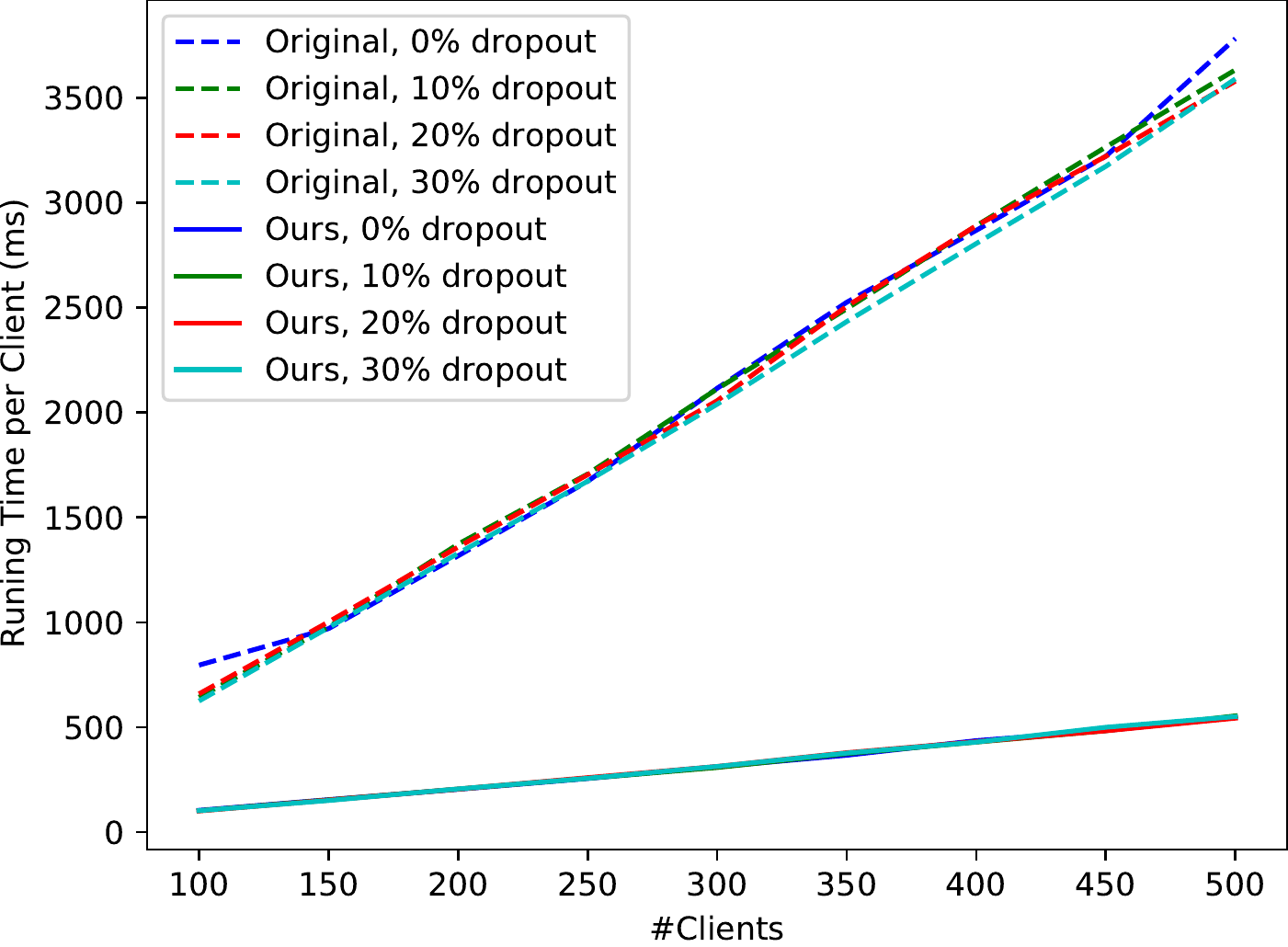}
		\caption{Running time per users, as the number of clients increases. The data vector size is fixed to 100K entries}
	\end{subfigure}%
	\begin{subfigure}[t]{0.30\textwidth}
		\captionsetup{width=0.9\textwidth}
		\centering
		\includegraphics[height=1.5in]{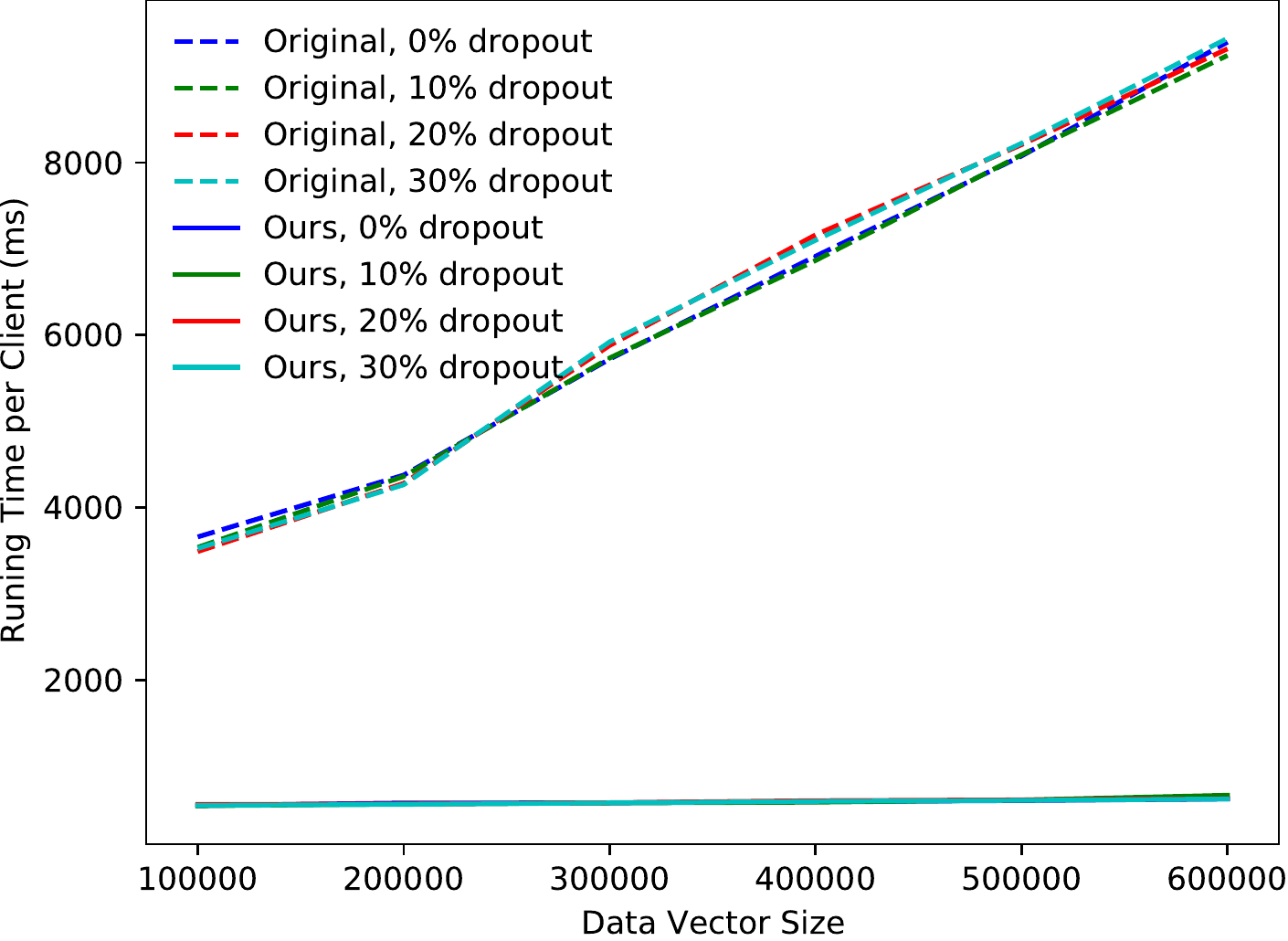}
		\caption{Running time per users, as the data vector size increases. The number of users is fixed to 500}
	\end{subfigure}%
	\begin{subfigure}[t]{0.30\textwidth}
		\captionsetup{width=0.9\textwidth}
		\centering
		\includegraphics[height=1.5in]{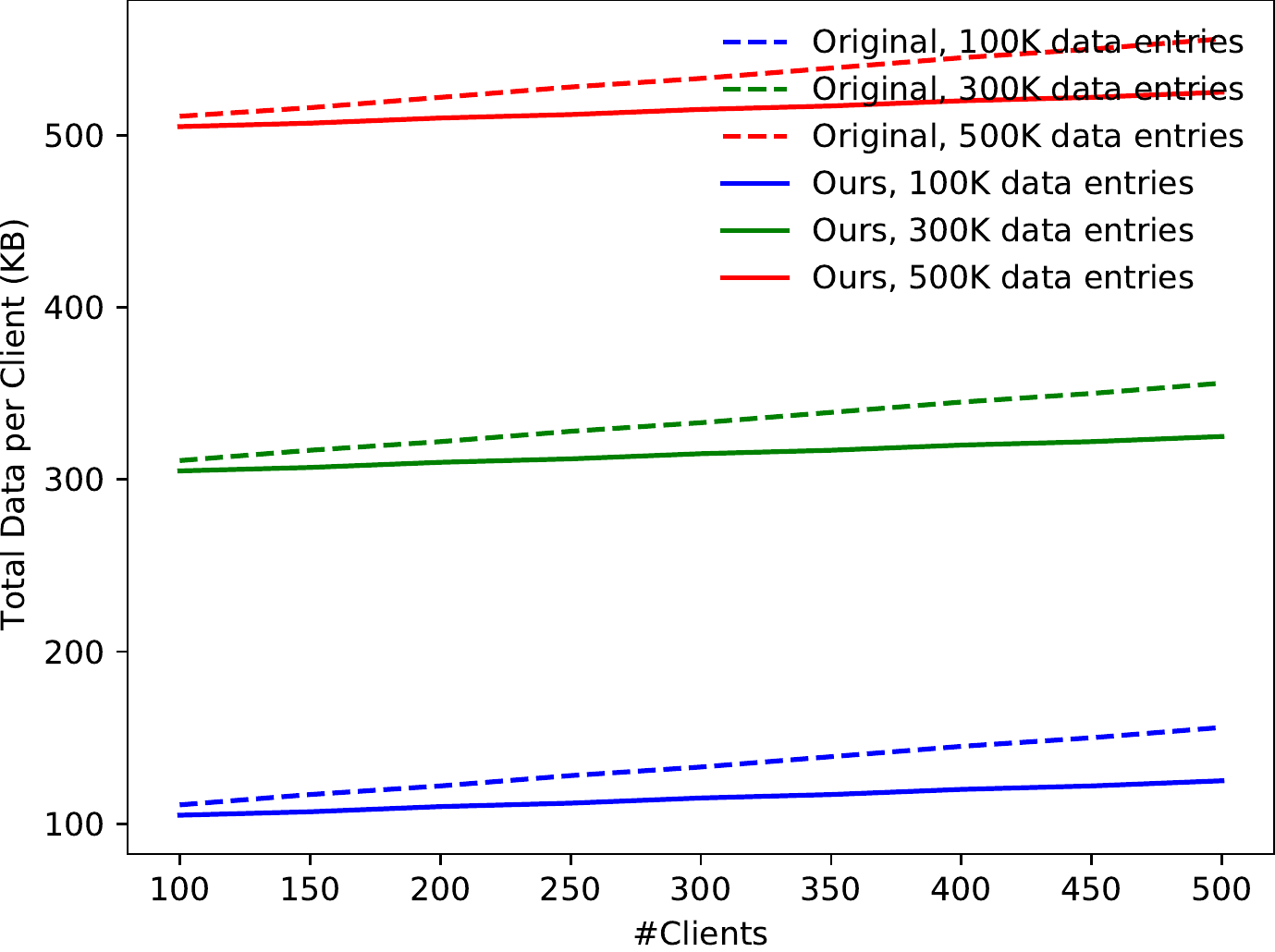}
		\caption{Total data transfer per client, as the number of clients increases. Assumes no dropouts.}
	\end{subfigure}
	\caption{Client Running Time. Solid line represent our Tree Based Secure Aggregation Protocol. Dash line represent original Secure Aggregation Protocol}\label{f:client}
\end{figure*}
\section{Performance analysis: Simulation}\label{apdx:performace}
To evaluate the performance of our protocol, we test our protocol in different tree structures with 1000 and 1500 participants under a $15\%$ drop rate.
Comparison of our protocol with the original secure aggregation protocol is presented in Appendix \ref{apdx:performace}. 
The experiments are performed on a Windows desktop with Intel i5-8400 (2.80 Ghz), with 16 GB of RAM. We assume that users will only drop out of the protocol after sending their secret shares to other users in the same subgroup, but before sending masked input to the server. This is the "worst-case" dropout because the server needs to perform mask recovery for other users in that subgroup. As shown in Table \ref{tab:tree}, when the heights or degree of the tree structure increase, the user's running time and communication costs decrease significantly. Therefore, by adjusting the tree structure, SAFELearning is able to control its complexity by limiting the size of subgroups, which makes it highly scalable.

We compare the overhead of our protocol with Bonawitz et. al. \cite{bonawitz2017practical} (abbr. as ``original") in the server and user side with simulation experiments. In the experiment,  we use $2\times 2$  tree structure and each user $u$ pair mask with previous and next $\kappa=4$ users.

Fig. \ref{f:server} shows the server's running time comparison between two protocols when the number of clients and data vector size increases. Note that, with 0 dropout rate, the server's running time is almost the same between two protocols. This is because as long as all pairwise masks are canceled out, the server of both protocols only needs to unmask $b_u$ for each survived user. But when there are the dropped users, the server in the original protocol needs to unmask $O(N)$ pairwise mask for each dropped user where our protocol only needs to unmask $O(1)$. So, in Fig. \ref{f:server}, the running time growth rate of the original protocol is increased significantly as the dropout rate rises up when our protocol still keeps the same growth rate as 0 dropout rate.

As seen in Fig. \ref{f:client} (a), The users' running time of both protocol increases linearly when the number of users increases. But our protocol has a significantly lower increase rate because the computation complexity of each user in our protocol is $O(N+n^2+m)$ instead of $O(N^2+mN)$. And because the number of pairwise masking needed in our protocol is constant for each user, the user's computation cost won't have significant change when the data vector size increase. So in Fig. \ref{f:client} (b), the running time's increase speed of our protocol is imperceptible compare to the original protocol that needs $N$ pairwise masking operations for each user. In Fig. \ref{f:client} (c), the plot shows that both protocols' communication cost are dominant by the data vector size, but our protocol still has lower communication expansion as the number of clients increase, although this increase is relatively small compared to the impact of increasing the size of the data vector.

\end{document}